\documentclass[runningheads,]{llncs}

\usepackage[utf8]{inputenc}
\usepackage{lmodern}

\usepackage{tikz}
\usetikzlibrary{arrows.meta, automata, positioning}

\usepackage{amsmath}
\usepackage{mathtools}
\usepackage{amssymb}
\usepackage{booktabs}
\usepackage{wrapfig}

\usepackage{algorithm}
\usepackage[noEnd]{algpseudocodex}

\usepackage{subcaption}
\usepackage[normalem]{ulem}

\usepackage{url}

\usepackage{fancybox}
\usepackage[capitalise,noabbrev]{cleveref}
\usepackage{todonotes}
\usepackage{cite}
\usepackage[inline]{enumitem}
\usepackage{fvextra}
\usepackage{comment}

\usepackage{framed}
\newcommand{\myparagraph}[1]{\smallskip\noindent \textbf{#1}}

\usepackage[inline]{enumitem}

\spnewtheorem{assumption}{Assumption}{\itshape}{\normalshape}%

\excludecomment{FinalVersion}
\includecomment{ArxivVersion}

\begin{FinalVersion}
  \newcommand{\FinalOrArxiv}[2]{#1}
\end{FinalVersion}

\begin{ArxivVersion}
  \newcommand{\FinalOrArxiv}[2]{#2}
\end{ArxivVersion}

\crefformat{section}{{\S}#2#1#3}
\Crefname{theorem}{Thm.}{Thm.}
\Crefname{propositioncnt}{Prop.}{Prop.}
\Crefname{lemmacnt}{Lem.}{Lem.}
\Crefname{corollarycnt}{Cor.}{Cor.}
\Crefname{factcnt}{Fact}{Fact}
\Crefname{remarkcnt}{Rem.}{Rem.}
\Crefname{notationcnt}{Notation}{Notation}

\Crefname{assumption}{Asm.}{Asm.}
\Crefname{problem}{Prob.}{Prob.}

\newcommand{\MatchAw}{Match${}_{\mathcal{A},w}$}
\newcommand{\MatchArw}{Match${}_{\mathcal{A}(r),w}$}

\newcommand{\MatchLaAt}{Match$\mathrm{\text{-}(la,at)}$}
\newcommand{\MatchLaAtAw}{Match$\mathrm{\text{-}(la,at)}_{\mathcal{A},w}$}
\newcommand{\MatchLaAtApw}{Match$\mathrm{\text{-}(la,at)}_{\mathcal{A}',w}$}
\newcommand{\MatchLaAtArw}{Match$\mathrm{\text{-}(la,at)}_{\mathcal{A}(r),w}$}

\newcommand{\DavisSLMAw}{DavisSL${}^M_{\mathcal{A}, w}$}
\newcommand{\DavisSLMzAw}{DavisSL${}^{M_0}_{\mathcal{A}, w}$}
\newcommand{\DavisSLImplMAw}{DavisSLImpl${}^M_{\mathcal{A}, w}$}
\newcommand{\MemoMAw}{Memo${}^M_{\mathcal{A}, w}$}
\newcommand{\MemoMzAw}{Memo${}^{M_0}_{\mathcal{A}, w}$}
\newcommand{\MemoExitLa}{MemoExit$\mathrm{\text{-}la}$}
\newcommand{\MemoExitLaMzAw}{MemoExit$\mathrm{\text{-}la}^{M_0}_{\mathcal{A}, w}$}
\newcommand{\MemoEnterLa}{MemoEnter$\mathrm{\text{-}la}$}
\newcommand{\MemoEnterLaMzAw}{MemoEnter$\mathrm{\text{-}la}^{M_0}_{\mathcal{A}, w}$}
\newcommand{\MemoLa}{Memo$\mathrm{\text{-}la}$}
\newcommand{\MemoLaMAw}{Memo$\mathrm{\text{-}la}^M_{\mathcal{A}, w}$}
\newcommand{\MemoLaMApw}{Memo$\mathrm{\text{-}la}^M_{\mathcal{A}', w}$}
\newcommand{\MemoLaMzAw}{Memo$\mathrm{\text{-}la}^{M_0}_{\mathcal{A}, w}$}
\newcommand{\MemoEnterAt}{MemoEnter$\mathrm{\text{-}at}$}
\newcommand{\MemoExitAt}{MemoExit$\mathrm{\text{-}at}$}
\newcommand{\MemoEnterAtMzAw}{MemoEnter$\mathrm{\text{-}at}^{M_0}_{\mathcal{A}, w}$}
\newcommand{\MemoAt}{Memo$\mathrm{\text{-}at}$}
\newcommand{\MemoAtMAzAw}{Memo$\mathrm{\text{-}at}^M_{\mathcal{A}_0, \mathcal{A}, w}$}
\newcommand{\MemoAtMAzApw}{Memo$\mathrm{\text{-}at}^M_{\mathcal{A}_0, \mathcal{A}', w}$}
\newcommand{\MemoAtMzAAw}{Memo$\mathrm{\text{-}at}^{M_0}_{\mathcal{A}, \mathcal{A}, w}$}
\newcommand{\MemoAst}{Memo$\mathrm{\text{-}\ast}^{M_0}_{\mathcal{A}, \mathcal{A}, w}$}
\newcommand{\MemoAstMAw}{Memo$\mathrm{\text{-}\ast}^M_{\mathcal{A}, w}$}
\newcommand{\MemoAstMzAw}{Memo$\mathrm{\text{-}\ast}^{M_0}_{\mathcal{A}, w}$}
\newcommand{\MemoLaAtMAzAw}{Memo$\mathrm{\text{-}(la,at)}^M_{\mathcal{A}_0, \mathcal{A}, w}$}
\newcommand{\MemoLaAtMApApw}{Memo$\mathrm{\text{-}(la,at)}^M_{\mathcal{A}', \mathcal{A}', w}$}
\newcommand{\MemoLaAtMAzApw}{Memo$\mathrm{\text{-}(la,at)}^M_{\mathcal{A}_0, \mathcal{A}', w}$}
\newcommand{\MemoLaAtMzAAw}{Memo$\mathrm{\text{-}(la,at)}^{M_0}_{\mathcal{A}, \mathcal{A}, w}$}

\begin{document}

\title{Efficient Matching with Memoization for Regexes with Look-around and Atomic Grouping\FinalOrArxiv{}{ (Extended Version)}\thanks{The authors are supported by  CREST ZT-IoT Project (No.\ JPMJCR21M3), 
 ERATO HASUO Metamathematics for Systems
Design Project (No.\ JPMJER1603), and
 ASPIRE Grant No.\ JPMJAP2301, JST.
}}
\titlerunning{Efficient Matching with Memoization for (la, at)-regexes}

%

\author{Hiroya Fujinami\inst{1,2}\orcidID{0009-0007-0794-5743} \and Ichiro Hasuo\inst{1,2}\orcidID{0000-0002-8300-4650}}
\authorrunning{H. Fujinami and I. Hasuo}

\institute{
  National Institute of Informatics, Tokyo, Japan
  \and
  SOKENDAI (The Graduate University for Advanced Studies), Kanagawa, Japan
}

\maketitle

\begin{abstract}


\emph{Regular expression (regex) matching} is fundamental in many applications, especially in web services. However, matching by \emph{backtracking}---preferred by most real-world implementations for its practical performance and backward compatibility---can suffer from so-called \emph{catastrophic backtracking}, which makes the number of backtracking super-linear and leads to the well-known ReDoS vulnerability.
Inspired by a recent algorithm by Davis et al.\ 
that runs in linear time for (non-extended) regexes, we study efficient backtracking matching for regexes with two common extensions, namely \emph{look-around} and \emph{atomic grouping}. 
We present linear-time backtracking matching algorithms for these extended regexes. Their efficiency relies on \emph{memoization}, much like the one by Davis et al.; we also strive for smaller memoization tables by carefully trimming their range.  Our experiments---we used some real-world regexes with the aforementioned extensions---confirm the performance advantage of our algorithms.

  
\end{abstract}
 
\section{Introduction}\label{sec:intro}

\myparagraph{Regex Matching}
\emph{Regular expressions} (\emph{regexes}) are a fundamental formalism for various pattern-matching tasks.
Many regex matching implementations, however, suffer from occasional super-linear growth of their execution time. Such excessive execution time can be exploited for DoS attacks---this is a vulnerability called \emph{regex denial of service} (\emph{ReDoS}). ReDoS is recognized as a significant security concern in many real-world systems, especially web services such as Stack Overflow and Cloudflare (see \cref{ssec:catastrophic} for more details). 

\myparagraph{Need for Efficient Backtracking Regex Matching}
The principal cause of ReDoS is \emph{catastrophic backtracking}, that is, the explosion of recursion in a backtracking-based matching algorithm. 

In regex matching, in general, a regex $r$ is converted into a non-deterministic finite automaton (NFA) $\mathcal{A}$, and the latter is executed for an input string $w$. The non-determinism of $\mathcal{A}$ can be resolved in either a \emph{depth-first} or a \emph{breadth-first} manner. The former is called \emph{backtracking regex matching}, and the latter is the \emph{on-the-fly DFA construction}. 

Catastrophic backtracking and ReDoS are phenomena unique to the former (i.e., backtracking)---as is well-known, the time complexity of the on-the-fly DFA construction is linear (i.e., $O(|w|)$). Indeed, many modern regex implementations are based on the on-the-fly DFA construction, including RE2\footnote{\url{https://github.com/google/re2}}, Go's \texttt{regexp}\footnote{\url{https://pkg.go.dev/regexp}}, and Rust's \texttt{regex}\footnote{\url{https://docs.rs/regex/latest/regex/}}.

It is practically essential, however, to make backtracking regex matching more efficient. A principal reason is \emph{consistency}. Most existing regex matching implementations use backtracking, and they return only one matching position out of many (see \cref{ssec:backtracking}). While it is possible to replace them with on-the-fly DFA matching, it is non-trivial to ensure consistency, that is, that the chosen matching position is the same as the original backtracking matching implementation.
.NET's regex implementation has a linear-time complexity backend using a derivative-based approach, which is compatible with a backtracking backend. Still, it does not support look-around and atomic grouping~\cite{DBLP:journals/pacmpl/MoseleyNRSTVWX23}.
Once the returned matching position changes, it can unexpectedly affect the behavior of all the systems (e.g., web services) that use regex matching.

Another reason for improving backtracking regex matching is its \emph{extensibility}. There are many extensions of regexes widely used---such as the ones we study, namely look-around and atomic grouping---and they are supported by few on-the-fly DFA matching implementations.

\myparagraph{Existing Work: Linear-time Backtracking Matching with Memoization}
\emph{Memoization} is a well-known technique for speeding up recursive computations.
The recent work~\cite{DBLP:conf/sp/DavisSL21} shows that memoization can be applied to backtracking regex matching with consistency in mind. Specifically, the work~\cite{DBLP:conf/sp/DavisSL21} presents a backtracking matching algorithm that runs in $O(|w|)$ time---thus, it is theoretically guaranteed to avoid catastrophic backtracking---for regexes without extensions. (They also mention application to extended regexes in~\cite{DBLP:conf/sp/DavisSL21}, but we found issues in their discussion---see \cref{rem:davis2021Err}).

\myparagraph{Our Contribution: Linear-time Backtracking Matching for Some Extended Regexes}
In this paper, we present a linear-time backtracking matching algorithm for regexes with \emph{look-around} and \emph{atomic grouping}, two real-world extensions of regexes. It uses memoization in order to achieve a linear-time complexity. We also prove that it is consistent (i.e., it chooses the same matching position as the original algorithm without memoization). 

The technical key to our algorithm is the design of suitable memoization tables. We follow the general idea in~\cite{DBLP:conf/sp/DavisSL21} of using memoization for backtracking matching, but our examination of its issues with extended regexes (\cref{rem:davis2021Err}) shows that the range---i.e., the set of possible entries---of memoization tables should be suitably extended. Specifically, the range in~\cite{DBLP:conf/sp/DavisSL21} is $\{\mathbf{false}\}$, recording only matching failures; it is extended in our algorithm to $\{ \mathsf{Failure}(j)\mid j \in \{ 0, \dots, \nu(\mathcal{A}) \} \}\cup\{\mathsf{Success}\}$.
Here, $\nu(\mathcal{A})$ is the maximum nesting depth of atomic grouping for the (extended) NFA $\mathcal{A}$,  defined in \cref{sec:memo-at}.

Our development is rigorous and systematic, based on the notion of 
NFA whose labels can themselves be NFAs. This extended notion of NFA is suggested in~\cite[Section~IX.B]{DBLP:conf/sp/DavisSL21}; in this paper, we formalize it and build its theory.

We experimentally evaluate our algorithm; the experiment results confirm its performance advantages. Additionally, we survey the usage status of look-around and atomic grouping---two regex extensions of our interest---in real-world regexes and demonstrate their wide usage (\cref{sec:exp}).

\myparagraph{Technical Contributions}
We summarize our technical contributions.
\begin{itemize}
  \item We propose a backtracking matching algorithm for regexes with look-around, proving its linear-time complexity (\cref{sec:memo-la}). This algorithm fixes the issues in the algorithm in~\cite{DBLP:conf/sp/DavisSL21} (\cref{rem:davis2021Err}) and restores correctness and linearity.

  \item We also propose a  backtracking matching algorithm for regexes with atomic grouping, proving its linear-time complexity (\cref{sec:memo-la}).
  \item We experimentally confirm the performance of our algorithms (\cref{sec:exp}).
  \item We investigate the usage status of look-around and atomic grouping in real-world regexes and confirm their wide usage (\cref{sec:exp}). 
  \item We establish a rigorous theoretical basis for our algorithms for extended regexes, namely NFAs with sub-automata (\cref{ssec:sub-automata}). 
\end{itemize}

\myparagraph{Organization}
We provide some preliminaries in \cref{sec:prelim}, such as regex extensions of our interest. Our formalization of NFAs with sub-automata is also presented here.
In \cref{sec:prev-work},
we discuss the work~\cite{DBLP:conf/sp/DavisSL21} that is closest to ours.
We present our matching algorithm for regex with look-around in \cref{sec:memo-la} and the one for regex with atomic grouping in \cref{sec:memo-at}.
Then, we discuss our implementation and experimental evaluation in \cref{sec:exp}.
We conclude in \cref{sec:conclusion}.

\FinalOrArxiv{Some additional proofs and other materials are deferred to the appendices in the extended version~\cite{FujinamiH24ESOP_arxiv_extended_ver}.}{}

\myparagraph{Related Work}
Many related works are discussed elsewhere in suitable contexts. Here, we discuss other related works.

There are many theoretical studies on look-around and atomic grouping.
The work~\cite{morihata2012} 
is a theoretical study of look-ahead operators; it shows
 how to convert them to finite automata. Another conversion based on derivatives is introduced in~\cite{DBLP:journals/jip/MiyazakiM19}.
The work~\cite{DBLP:journals/jucs/BerglundML21} conducts a fine-grained analysis of the size of DFAs obtained from converting regexes with look-ahead, improving the bounds given in~\cite{morihata2012,DBLP:journals/jip/MiyazakiM19}. 
The work~\cite{DBLP:conf/fscd/ChidaT22} 
discusses the relation between look-ahead operators and back-references in regexes.
A recent study~\cite{MamourasC24} presents a linear-time matching algorithm for regexes with look-around; it uses a memoization-like construct for efficiency. However, the compatibility with backtracking is not a concern there, unlike the current work.
%
On atomic grouping, conversion to finite automata is proposed~\cite{DBLP:conf/wia/BerglundMWW17}, where atomic grouping is simulated by look-ahead.


Another common regex extension is \emph{back-reference}.
We do not deal with this extension because
\begin{enumerate*}[label=\arabic*)]
  \item this extension is known to be non-regular (i.e., the language class defined by back-reference is beyond regular), and
  \item its matching problem is known to be NP-complete~\cite{DBLP:books/el/leeuwen90/Aho90} (thus the search for linear-time matching is doomed).
\end{enumerate*}
There are other extensions (absent operators, conditional branching, etc.), but they are  used less often (cf.~\cref{sec:exp}).

ReDoS countermeasures are an active scientific topic.
Besides efficient matching, there are two directions for them: \emph{ReDoS detection} and \emph{ReDoS repair}.
ReDoS detection is a problem that determines whether a given regex can cause catastrophic backtracking.
This can be done by finding specific structures in a transition diagram of an automaton~\cite{DBLP:conf/nss/KirrageRT13,DBLP:journals/corr/BerglundDM14,DBLP:journals/corr/RathnayakeT14,sugiyama2014,DBLP:conf/wia/WeidemanMBW16,DBLP:conf/tacas/WustholzOHD17}.
Besides, dynamic analysis, such as fuzzing~\cite{DBLP:conf/kbse/Shen000ML18}, and combinations of static and dynamic analyses~\cite{DBLP:conf/uss/LiCC0PCCC21} are studied.
ReDoS repair is a problem of modifying a given regex so that it does not cause ReDoS.
Known solutions include exploring ReDoS-free regexes using SMT solvers~\cite{DBLP:conf/kbse/LiXCCGCZ20,DBLP:conf/sp/ChidaT22} and rule-based rewriting of vulnerable regexes~\cite{DBLP:conf/uss/LiS0CLLCC0X22}.
These 
ReDoS detection and repair measures are computationally demanding, and their real-world deployment is limited.

There are other implementation-level studies on speeding up regex matching, such as Just-in-Time (JIT) compilation~\cite{DBLP:conf/cgo/Herczeg14} and FPGA~\cite{DBLP:conf/fccm/SidhuP01}.
However, these studies are not intended to prevent catastrophic backtracking. 

\section{Preliminaries}\label{sec:prelim}


We introduce preliminaries for this paper. Firstly, we present some basic concepts such as regexes, NFAs, conversion from regexes to NFAs, and backtracking matching.  We then discuss \emph{catastrophic backtracking} and the \emph{ReDoS vulnerability} that it can cause. Finally, we introduce \emph{look-around} and \emph{atomic grouping} as practical regex extensions and \emph{NFAs with sub-automata} for these extensions.


We fix a finite set $\Sigma$ as an alphabet throughout this paper.
We call sequences of elements of $\Sigma$ \emph{strings}.
The \emph{empty string} is denoted by $\varepsilon$.
For a string $w = \sigma_0 \sigma_1 \dots \sigma_{n-1}$, the \emph{length} of $w$, denoted by $|w|$, is defined as $|w| = n$.
We also write $w[i] = \sigma_i$ for $i \in \{ 0, \dots, n - 1 \}$.

We use partial functions for memoization.
For two sets $A$ and $B$, a partial function $G$ from $A$ to $B$, denoted by $G\colon A \rightharpoonup B$, is defined as a function $G\colon A \to B \cup \{ \bot \}$. Here $\bot$ is the element for  ``undefined,'' and it is assumed that $\bot\not\in B$.

Let 
$G\colon A \rightharpoonup B$ be a partial function, $a \in A$, and $b \in B$. We let $G(a) \gets b$ denote an updated partial function: it carries $a$ to $b$, and any other $x\in A$ to $G(x)$ (it is undefined if $G(x)$ is initially undefined). 

\subsection{Regexes}\label{ssec:regexes}

\emph{Regular expressions} (\emph{regexes}) are defined by the following abstract grammar.
\begin{align*}
\begin{array}{rlllllllllllllll}
   r \Coloneqq&\ 
    \sigma && \text{(a (literal) character, where $\sigma \in \Sigma$)} \quad &
    |&\ 
    \varepsilon && \text{(the empty string)} \\
    |&\ 
    r | r && \text{(an alternation)} \quad &
    |&\ 
    r\!\cdot\!r && \text{(a concatenation)} \\
    |&\ 
    r^\ast && \text{(a repetition)}
\end{array}
\end{align*}

The concatenation operator $\cdot$ may be omitted when there is no ambiguity.
The precedence of operators is as follows: repetition, concatenation, and alternation.
For example, $ab^\ast | c$ means $(a \cdot (b^\ast)) | c$.

For a regex $r$, the \emph{size of $r$}, denoted by $|r|$, is defined as follows: $|\sigma| = |\varepsilon| = 1$, $|(r_1|r_2)| = |r_1 \cdot r_2| = |r_1| + |r_2| + 1$, and $|r^\ast| = |r| + 1$.

\subsection{NFAs}\label{ssec:nfa}

A \emph{non-deterministic finite state automaton} (\emph{NFA}) is a quadruple $(Q, q_0, F, T)$, where $Q$ is a finite set of states $q_0 \in Q$ is an initial state, $F \subseteq Q$ is a set of accepting states, and $T$ is a transition function. For each $q \in Q \setminus F$, $T(q)$ can be one of the following: $T(q) = \mathsf{Eps}({q'})$, $T(q) = \mathsf{Branch}(q', q'')$, and $T(q) = \mathsf{Char}(\sigma, {q'})$ where $q', q'' \in Q$ and $\sigma \in \Sigma$.

The above definition of a transition function $T$ is tailored to our purpose of backtracking.
Compared to the common definition $\delta\colon Q \times (\{ \varepsilon \} \cup \Sigma) \to 2^Q$, it expresses general branching as combinations of certain \emph{elementary branchings}.
The latter is namely one transition by $\varepsilon$, two transitions by $\varepsilon$, and one transition by a certain character $\sigma \in \Sigma$.
This makes the description of backtracking matching easier.
Note, in particular, that the successors $q', q''$ in the branching $\mathsf{Branch}(q', q'')$ are ordered. Here, $q'$ and $q''$ are called the \emph{first} and \emph{second successors}, respectively.
This definition of transition functions is similar to the op-codes of many real-world regex-matching implementations (cf.~\cite{cox2009}).

\begin{figure}[tb]
  \centering
  \begin{minipage}[t]{0.24\textwidth}
    \centering
    \begin{tikzpicture}[initial text=, node distance=2cm, on grid,auto,scale=0.75]
      \scriptsize
      \node[state,initial] (q_0) at (0, 0) {$q_0$};
      \node[state,accepting] (q_1) at (2.25, 0) {$q_1$};
      \path[->] (q_0) edge node {$\mathsf{Char}(\sigma)$} (q_1);
    \end{tikzpicture}
    \subcaption{$\sigma\in\Sigma$}
  \end{minipage}
  \quad
  \begin{minipage}[t]{0.22\textwidth}
    \centering
    \begin{tikzpicture}[initial text=, node distance=2cm, on grid,auto,scale=0.75]
      \scriptsize
      \node[state,initial] (q_0) at (0, 0) {$q_0$};
      \node[state,accepting] (q_1) at (2, 0) {$q_1$};
      \path[->] (q_0) edge node {$\mathsf{Eps}$} (q_1);
    \end{tikzpicture}
    \subcaption{$\varepsilon$}
  \end{minipage}
  \begin{minipage}[t]{0.45\textwidth}
    \centering
    \begin{tikzpicture}[initial text=, node distance=2cm, on grid,auto,scale=0.75]
      \scriptsize
      \node[state,rectangle,minimum width=1cm] (r_1) at (1.8, 0) {$r_1$};
      \node[state,rectangle,minimum width=1cm] (r_2) at (4.1, 0) {$r_2$};
      \path[->] (0.8, 0) edge (r_1);
      \path[->] (r_1) edge node {$\mathsf{Eps}$} (r_2);
      \path[->] (r_2) edge (5.1, 0);
    \end{tikzpicture}
    \subcaption{$r_1 \cdot r_2$}
  \end{minipage}
    \begin{minipage}[t]{0.45\textwidth}
    \centering
    \begin{tikzpicture}[initial text=, node distance=2cm, on grid,auto,scale=0.75]
      \scriptsize
      \node[state,initial] (q_0) at (0, 0) {$q_0$};
      \node[state,rectangle,minimum width=1cm] (r_1) at (3, 1) {$r_1$};
      \node[state,rectangle,minimum width=1cm] (r_2) at (3, -1) {$r_2$};
      \node[state,accepting] (q_1) at (4.5, 0) {$q_1$};
      \path[-] (q_0) edge node {$\mathsf{Branch}$} (1.5, 0);
      \path[->] (1.5, 0) edge (r_1.west);
      \path[->] (1.5, 0) edge (r_2.west);
      \path[->] (r_1.east) edge (q_1);
      \path[->] (r_2.east) edge (q_1);
    \end{tikzpicture}
    \subcaption{$r_1 | r_2$}
  \end{minipage}
  \begin{minipage}[t]{0.45\textwidth}
    \centering
    \begin{tikzpicture}[initial text=, node distance=2cm, on grid,auto,scale=0.75]
      \scriptsize
      \node[state,initial] (q_0) at (0, 0) {$q_0$};
      \node[state,rectangle,minimum width=1cm] (r) at (3, 1) {$r$};
      \node[state,accepting] (q_1) at (2.5, -1) {$q_1$};
      \path[-] (q_0) edge node {$\mathsf{Branch}$} (1.5, 0);
      \path[->] (1.5, 0) edge (r.west);
      \path[-] (r.east) edge (4.5, 1);
      \path[->] (4.5, 1) edge [bend right=4cm] node [above] {$\mathsf{Eps}$} (q_0.north);
      \path[->] (1.5, 0) edge (q_1);
    \end{tikzpicture}
    \subcaption{$r^\ast$}\label{subfig:star}
  \end{minipage}
  \caption{a conversion from regexes to NFAs}\label{fig:nfa:conversion}
\end{figure}
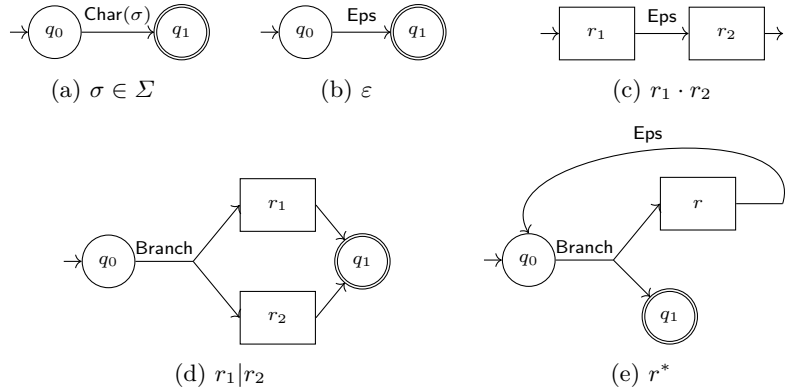

We present a conversion from regexes to NFAs (see \cref{fig:nfa:conversion}); it is similar to the Thompson--McNaughton--Yamada construction~\cite{DBLP:journals/tc/McNaughtonY60,DBLP:journals/cacm/Thompson68}.
For a regex $r$, 
 $\mathcal{A}(r)$ denotes
the \emph{NFA $\mathcal{A}$ converted from $r$}.
In the figure, labels on arrows show kinds of transitions.
In a $\mathsf{Branch}$ transition, the top arrow points to the first successor, and the bottom points to the second successor.
Rectangles indicate that the conversion is applied to sub-expressions inductively.
Because each case of this construction introduces at most two new states, for a regex $r$ and the NFA $\mathcal{A}(r) = (Q, q_0, F, T)$, we have $|Q| = O(|r|)$.

We collectively call $\mathsf{Eps}$ and $\mathsf{Branch}$ transitions \emph{$\varepsilon$-transitions}.
Later in this paper, if there are consecutive $\varepsilon$-transitions, they may be shown as a single transition in a figure.
When a certain state returns to itself by $\varepsilon$-transitions, such a sequence of $\varepsilon$-transitions is called an \emph{$\varepsilon$-loop}.
$\varepsilon$-loops are problematic in matching because they cause infinite loops in matching.

An $\varepsilon$-loop can be detected during matching by recording a position on an input string when a state is visited.
When an $\varepsilon$-loop is detected, several solutions exist to deal with it (see, e.g.,~\cite{DBLP:journals/japll/SakumaMV12}),
such as treating an $\varepsilon$-loop as a failure (e.g., JavaScript and RE2) or treating it as a success but escaping it (e.g., Perl).
These solutions can be easily adapted to our algorithms; therefore, for the simplicity of presentation, we introduce the following assumption.

\begin{assumption}[no $\varepsilon$-loops]\label{assum:no-eps-loops}
NFAs do not contain $\varepsilon$-loops.
\end{assumption}

\subsection{Backtracking Matching}\label{ssec:backtracking}

\begin{algorithm}[tb]
  \caption{a partial backtracking matching algorithm for NFAs}\label{alg:matching}
  \begin{algorithmic}[1]
    \Function{\MatchAw}{$q, i$}
     \Statex\hspace*{\algorithmicindent}
       \begin{tabular}{rl} 
        \textbf{Parameters}: & an NFA $\mathcal{A}$, and an input string $w$ \\
             \textbf{Input}: & a current state $q$, and a current position $i$ \\
            \textbf{Output}: & returns $\mathsf{SuccessAt}(i')$ if the matching succeeds, or \\
                             & returns $\mathsf{Failure}$ if the matching fails
       \end{tabular}
      \State$(Q, q_0, F, T) = \mathcal{A}$
      \If{$q \in F$}\label{line:matching-trans-begin}
        \Return$\mathsf{SuccessAt}(i)$
      \ElsIf{$T(q) = \mathsf{Eps}(q')$}
        \Return$\Call{\MatchAw}{q', i}$
      \ElsIf{$T(q) = \mathsf{Branch}(q', q'')$}\label{line:matching-branch-1}
        \State$\mathsf{result} \leftarrow \Call{\MatchAw}{q', i}$\label{line:matching-branch-2}
        \If{$\mathsf{result} = \mathsf{Failure}$}
          \Return$\Call{\MatchAw}{q'', i}$\label{line:matching-branch-3}
        \Else\ 
          \Return$\mathsf{result}$\label{line:matching-branch-4}
        \EndIf%
      \ElsIf{$T(q) = \mathsf{Char}(\sigma, q')$}
        \If{$i < |w|$ \textbf{and} $w[i] = \sigma$}
          \Return$\Call{\MatchAw}{q', i + 1}$
        \Else\ 
          \Return$\mathsf{Failure}$
        \EndIf%
      \EndIf\label{line:matching-trans-end}
    \EndFunction%
  \end{algorithmic}
\end{algorithm}

We present a basic backtracking matching algorithm for NFAs in \cref{alg:matching}. It serves as a basis for optimization by memoization, both in~\cite{DBLP:conf/sp/DavisSL21} and in the current work.

The function $\textproc{\MatchAw}$ is recursively called in this algorithm, but it must terminate on \cref{assum:no-eps-loops}.
It takes two parameters: $\mathcal{A}$ is an NFA, and $w$ is an input string.
It also takes two arguments: $q \in Q$ is the current state, and $i \in \{ 0, \dots, |w| \}$ is the current position on $w$.
$\Call{\MatchAw}{q_0, i}$ for an NFA $\mathcal{A} = (Q, q_0, F, T)$ returns $\mathsf{SuccessAt}(i')$ with the matching position $i' \in \{ 0, \dots, |w| \}$ if the matching with $\mathcal{A}$ succeeds from $i$ to $i'$ on $w$, or returns $\mathsf{Failure}$ if the matching fails.


The $\textproc{Match}$ function implements \emph{partial matching}: given the position $i\in\{0,\dotsc,|w|\}$ of interest, one obtains, by running $\Call{\MatchAw}{q_{0}, i}$,
 one ``matching position'' $i'$ (if it exists) such that $w[i]\, w[i+1]\,\dotsc w[i']$ is accepted by $\mathcal{A}$.
 Note the difference from \emph{total matching}: given $\mathcal{A}$ and $w$, it returns $\mathbf{true}$ if (the whole) $w$ is accepted by $\mathcal{A}$ and $\mathbf{false}$ otherwise.
The practical relevance of partial matching must be clear, as we can use it for text search and replacement.

\crefrange{line:matching-branch-1}{line:matching-branch-4} in \cref{alg:matching} perform matching for $\mathsf{Branch}$ transitions.
Here, the algorithm first tries matching from the first successor $q'$, and if that fails, it tries matching from the second successor $q''$ with the same position. 
This behavior is called \emph{backtracking}.


We define the \emph{regex partial matching} problem using the function $\textproc{Match}$.

\begin{problem}[regex partial matching]
\label{prob:regexParMatch}
\par\noindent
  \begin{tabular}{rl}
    \textbf{Input}: & a regex $r$, an input string $w$, and a starting position $i \in \{ 0, \dots, |w| \}$ \\
    \textbf{Output}: & returns $\Call{\MatchArw}{q_0, i}$ where $\mathcal{A}(r) = (Q, q_0, F, T)$.
  \end{tabular}
\end{problem}

\begin{remark}\label{rem:problemFormulation}
  One can say that the problem formulation is a bit strange.
  It requires, as output, a specific matching position chosen by a specific algorithm $\textproc{Match}$, while a usual formulation would require an arbitrary matching position.
  We take this formulation since we aim to show that our optimization by memoization not only solves partial matching but also is \emph{consistent} with an existing backtracking matching algorithm, in the sense we discussed in \cref{sec:intro}.
  We formulate consistency as correctness with respect to \cref{prob:regexParMatch}, that is, preserving the solution chosen by the specific algorithm $\textproc{Match}$.
  We also note that the algorithm $\textproc{Match}$ mirrors many existing implementations of regex matching (cf.\ \cref{ssec:nfa}).
\end{remark}

\subsection{Catastrophic Backtracking and ReDoS}\label{ssec:catastrophic}

In the execution of the $\textproc{Match}$ function (\cref{alg:matching}), depending on an NFA $\mathcal{A}$ and an input string $w$, the number of recursive calls for the $\textproc{Match}$ function may increase explosively, resulting in a very long matching time, as we will see in \cref{exa:catastrophic}.
This explosive increase in matching time is called \emph{catastrophic backtracking}.

\begin{figure}[tb]
  \centering
\scalebox{.7}{  \begin{tikzpicture}[initial text=, node distance=2cm, on grid,auto,scale=0.8]
    \scriptsize
    \node[state,initial] (q_0) at (0, 0) {$q_0$};
    \node[state] (q_1) at (3, 1) {$q_1$};
    \node[state] (q_2) at (6, 2) {$q_2$};
    \node[state] (q_3) at (8, 2) {$q_3$};
    \node[state] (q_4) at (6, 0) {$q_4$};
    \node[state] (q_5) at (8, 0) {$q_5$};
    \node[state] (q_6) at (10, 1) {$q_6$};
    \node[state] (q_7) at (3, -1) {$q_7$};
    \node[state,accepting] (q_8) at (5, -1) {$q_8$};
    \path[-] (q_0) edge node {$\mathsf{Branch}$} (1.5, 0);
    \path[->] (1.5, 0) edge (q_1.west);
    \path[-] (q_1) edge node {$\mathsf{Branch}$} (4.5, 1);
    \path[->] (4.5, 1) edge (q_2.west);
    \path[->] (4.5, 1) edge (q_4.west);
    \path[->] (q_2) edge node {$\mathsf{Char}(a)$} (q_3);
    \path[->] (q_4) edge node {$\mathsf{Char}(a)$} (q_5);
    \path[->] (q_3.east) edge node [above] {$\mathsf{Eps}$} (q_6);
    \path[->] (q_5.east) edge node [below] {$\mathsf{Eps}$} (q_6);
    \path[-] (q_6.east) edge (11, 1);
    \path[->] (11, 1) edge [bend right=4cm] node [below] {$\mathsf{Eps}$} (q_0.north);
    \path[->] (1.5, 0) edge (q_7);
    \path[->] (q_7) edge node {$\mathsf{Char}(b)$} (q_8);
  \end{tikzpicture}}
  \caption{the NFA $\mathcal{A}({(a|a)}^\ast b)$}\label{fig:nfa:aorastar}
\end{figure}
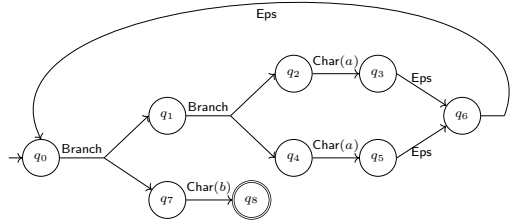

\begin{example}[catastrophic backtracking]\label{exa:catastrophic}
  Consider the NFA $\mathcal{A} = \mathcal{A}({(a|a)}^\ast b) = (Q, q_0, F, T)$ shown in \cref{fig:nfa:aorastar},  and let $w = \texttt{"$a^n c$"}$ (the string repeating $a$ of $n$ times and ending with $c$) be an input string.
  $\Call{\MatchAw}{q_0, 0}$ invokes recursive calls $O(2^n)$ times until returning $\mathsf{Failure}$.
  The reason for this recursive call explosion is to try all combinations on $q_2$ to $q_3$ and $q_4$ to $q_5$ transitions for each $a$ in $w$ during the matching.
\end{example}


\emph{Regexes denial of service} (\emph{ReDoS}) is a security vulnerability caused by catastrophic backtracking.
In ReDoS,
catastrophic backtracking causes a huge load on servers, making them unable to respond in a timely manner.
There are cases of service outages due to ReDoS at Stack Overflow in 2016~\cite{stackexchange2016} and at Cloudflare in 2019~\cite{cloudflare2019}.
Additionally, a 2018 study~\cite{DBLP:conf/uss/StaicuP18} reported that over 300 web services have potential ReDoS vulnerabilities.
Thus, ReDoS is a widespread problem in the real world, and there is a great need for countermeasures.

According to a 2019 study~\cite{DBLP:conf/kbse/MichaelDDLS19}, only 38\% of developers are aware of ReDoS.
This study also found that many developers find it difficult not only to read regexes but also to find and validate regexes to match their desires.
It is mentioned in~\cite{DBLP:conf/kbse/MichaelDDLS19} that developers use Internet resources such as Stack Overflow to find regexes.
In recent years, it has also become common to use generative AIs such as ChatGPT for such a purpose.
However, when the authors asked, ``Please suggest 10 regexes for validating email addresses'' to ChatGPT,\footnote{We used ChatGPT 3.5 (September 25, 2023 version).} 2 of the 10 suggested regexes would cause ReDoS (see \cref{tab:chatgpt-regexes}).
Developers may unknowingly use such vulnerable regexes.
For this reason, it is important to develop ReDoS countermeasures that can be achieved without the developer being aware of them. 

\begin{table}[tb]
  \caption{the regexes given by ChatGPT for the question ``Please suggest 10 regexes for validating email addresses''.\protect\footnotemark}\label{tab:chatgpt-regexes}
  \centering
  \footnotesize
  \begin{tabular}{lc}\toprule
    suggested regex (ChatGPT's comment) & vulnerable? \\
    \midrule
    \begin{tabular}{l}
      \EscVerb{^[a-zA-Z0-9._\%+-]+@[a-zA-Z0-9.-]+\\.[a-zA-Z]{2,}\$} \\ 
      \scriptsize\qquad(Basic Email Validation)
    \end{tabular} &  \\
    \begin{tabular}{l}
      \EscVerb{^[a-zA-Z0-9._\%+-]+@[a-zA-Z0-9.-]+\\.[a-zA-Z]{2,4}\$} \\ 
      \scriptsize\qquad(Basic Email Validation with TLD)
    \end{tabular} &  \\
    \begin{tabular}{l}
      \EscVerb{^[a-zA-Z0-9._\%+-]+@[a-zA-Z0-9.-]+\\.[a-zA-Z]{2,4}\$} \\
      \scriptsize\qquad(Strict Email Validation)
    \end{tabular} &  \\
    \begin{tabular}{l}
      \EscVerb{^[a-zA-Z0-9._\%+-]+@[a-zA-Z0-9.-]+\\.[a-zA-Z]{2,}(?:\\.[a-z} \\
      \EscVerb{A-Z]\{2,\})?\$} \\
      \scriptsize\qquad(Email Validation Allowing for Subdomains)
    \end{tabular} &  \\
    \begin{tabular}{l}
      \EscVerb{^[a-zA-Z0-9._\%+-]+@[a-zA-Z0-9.-]+\\.[a-zA-Z]{1,}\$} \\
      \scriptsize\qquad(Email Validation Allowing Single-Character Domain Name)
    \end{tabular} &  \\
    \begin{tabular}{l}
      \EscVerb{^[a-zA-Z0-9._\%+-]+@[a-zA-Z0-9.-]+\\.[a-zA-Z]{2,}(?:\\.[\\p\{} \\
      \EscVerb{L\}\\p{N}]{2,})?\$} \\
      \scriptsize\qquad(Email Validation Allowing Internationalized Domain Names (IDNs))
    \end{tabular} &  \\
    \begin{tabular}{l}
      \EscVerb{^(?:"[\\w\\s]+")?([a-zA-Z0-9._\%+-]+)@[a-zA-Z0-9.-]+\\.[a-zA} \\
      \EscVerb{-Z]{2,}\$} \\
      \scriptsize\qquad(Email Validation with Optional Quoted Local Part)
    \end{tabular} &  \\
    \begin{tabular}{l}
      \textcolor{red}{\EscVerb{^(?:\\([^()]*\\)|[\\w\\s]+)?([a-zA-Z0-9._\%+-]+)@[a-zA-Z0-9.-}} \\
      \textcolor{red}{\EscVerb{]+\\.[a-zA-Z]{2,}\$}} \\
      \scriptsize\qquad\textcolor{red}{(Email Validation with Optional Comments)}
    \end{tabular} & \textcolor{red}{vulnerable} \\
    \begin{tabular}{l}
      \EscVerb{^[a-zA-Z0-9._\%+-]+@[a-zA-Z0-9.-]+\\.[a-zA-Z]{2,4}\$i} \\
      \scriptsize\qquad(Email Validation Allowing for Case-Insensitive Domain)
    \end{tabular} &  \\
    \begin{tabular}{l}
      \textcolor{red}{\EscVerb{^(?:(?:[a-zA-Z0-9._\%+-]+@[a-zA-Z0-9.-]+\\.[a-zA-Z]{2,})|(}} \\
      \textcolor{red}{\EscVerb{[a-zA-Z0-9._\%+-]+)\\+[^@]+@[a-zA-Z0-9.-]+\\.[a-zA-Z]{2,})\$}} \\
      \scriptsize\qquad\textcolor{red}{(Email Validation with Support for Subaddressing)}
    \end{tabular} & \textcolor{red}{vulnerable} \\ \bottomrule
  \end{tabular}
\end{table}

Matching speed-up is a way to avoid causing ReDoS by ensuring that matching is linear in time to the length of an input string, freeing developers from worrying about ReDoS.
A popular method for matching speed-up is using breadth-first search for non-deterministic transition instead of backtracking (depth-first search); it is called the \emph{on-the-fly DFA construction}~\cite{cox2007,DBLP:journals/pacmpl/MoseleyNRSTVWX23}.
However, since look-around and atomic grouping are extensions based on backtracking (see \cref{ssec:sub-automata}), it is not obvious that they can be supported by the on-the-fly DFA construction.

\emph{Memoization} is another approach to ensuring linear-time backtracking matching; we pursue it in this paper.

\subsection{Regex Extensions: Look-around and Atomic Grouping}\label{ssec:practical-regex}



Many real-world regexes come with various extensions for enhanced expressivity~\cite{DBLP:books/daglib/0016809}. In this paper, we are interested in two classes of extensions, namely
\emph{look-around} and \emph{atomic grouping}.

\footnotetext{The second and third regexes are the same;  they are the actual output of ChatGPT.}

\myparagraph{Look-around}
Look-around
is a regex extension that allows constraints on 
 strings around a certain position.
It is also called \emph{zero-width assertion} (e.g., in~\cite{DBLP:conf/sp/DavisSL21}) because it does not consume any characters.
Look-around consists of four types: \emph{positive} or \emph{negative}, and \emph{look-ahead} or \emph{look-behind}.

Positive look-ahead is typically represented by the syntax \texttt{(?=$r$)}; its matching succeeds when, reading ahead from the current position of the input string, the matching of the inner regex $r$ succeeds. Note that the position for the overall matching does not change by the inner matching of $r$.
For example, the regex \texttt{/(?=bc)/} matches the string \texttt{"abc"} from position $1$ (i.e., after the first character \texttt{a}) without consuming any characters.

The matching of a negative look-ahead \texttt{(?!$r$)} succeeds when
the inner regex $r$ is \emph{not} matched. 

Positive or negative \emph{look-behind}---denoted by \texttt{(?<=$r$)} or \texttt{(?<!$r$)}, respectively---is similar to the above, with the difference that the inner matching of $r$ is performed \emph{backward}, i.e., from right to left.
For example, the regex \texttt{/(?<=ab)/} matches the string \texttt{"abc"} from position $2$ (i.e., before the last character \texttt{c}) without consuming any characters.

A typical use of look-around is to put a look-behind before (or a look-ahead after) a regex $r$. This is useful when one wants to perform a search or replacement of $r$ for only those occurrences that are in a certain context.
For example, the regex \EscVerb{/(?<=<p>)[^<]*(?=<\\/p>)/} matches only contents of the HTML \texttt{<p>} tag.
As another example, common assertions such as 
 \EscVerb{\\A} 
(this matches the beginning of a string) 
and  \EscVerb{\\z} (this matches the end) can be expressed using look-around, namely $\text{\EscVerb{\\A}} = \text{\EscVerb{(?<!.)}}$ and $\text{\EscVerb{\\z}} = \text{\EscVerb{(?!.)}}$.

\myparagraph{Atomic Grouping}
Atomic grouping
is a regex extension that controls backtracking behaviors.
It is designed to manually avoid problems caused by backtracking, such as catastrophic backtracking (\cref{ssec:catastrophic}).

Atomic grouping is represented by the syntax \texttt{(?>$r$)}; once the matching of the inner regex $r$ succeeds, the remaining branches in potential backtracking for matching $r$ are discarded.
For example, the regex \texttt{/(a|ab)c/} matches the string \texttt{"abc"}, but the regex \texttt{/(?>(a|ab))c/} using atomic grouping does \emph{not} match it.
This is because, once \texttt{a} in the atomic grouping matches the first character \texttt{a} of \texttt{"abc"}, the remaining branch \texttt{ab} (in \texttt{a|ab}) is discarded, and one is left with the regex \texttt{c} and the string \texttt{"bc"}.

Atomic grouping is often used for the purpose of preventing catastrophic backtracking.
In that case, it is used in combination with the repetition syntax, e.g., \texttt{(?>($r$*))} (often abbreviated as \texttt{$r$*+}) and \texttt{(?>($r$+))}  (abbrev. as \texttt{$r$++}).
These abbreviations are called \emph{possessive quantifiers}.
The former (namely \texttt{(?>($r$*))}) is intuitively understood as \texttt{(?>($\varepsilon$|$r$|$rr$|$\dotsc$))}, with the difference that longer matching is preferred (this is because the $\mathsf{Eps}$ loop is the first successor in Figure~\ref{subfig:star}). Once a longer match is found, the remaining branches (i.e., those for shorter matches) get discarded, thus preventing catastrophic backtracking. 

One might wonder if our (linear-time and thus ReDoS-free) matching algorithm should support atomic grouping---the principal use of atomic grouping is to suppress backtracking and avoid ReDoS. We do need to support it since, as we discussed in \cref{sec:intro}, ours is meant to be a drop-in replacement for matching implementations that are currently used.


\myparagraph{Our Target Extended Regexes} Our target class, namely
\emph{regexes with look-around and atomic grouping}, is defined by the following grammar.
\begin{align*}
  r \Coloneqq&\ \dots && \text{(the same as the regexes definition, \cref{ssec:regexes})} \\
  |&\ ( \texttt{?=} r )\ |\ ( \texttt{?!} r ) && \text{(positive and negative look-ahead)} \\
  |&\ ( \texttt{?<=} r )\ |\ ( \texttt{?<!} r ) && \text{(positive and negative look-behind)} \\
  |&\ ( \texttt{?>} r ) && \text{(atomic grouping)} 
\end{align*}
For brevity, 
we sometimes
refer to regexes with look-around and atomic grouping as \emph{(la, at)-regexes}.
We also refer to regexes with look-around as \emph{la-regexes} and regexes with atomic grouping as \emph{at-regexes}.

For a (la, at)-regex $r$, the size of $r$, denoted by $|r|$, is defined as the same as the regex one except for $|( \texttt{?=} r )| = |( \texttt{?>} r) | = |r| + 1$.

Look-around is known to be \emph{regular}: they can be converted to DFA, and the language family of la-regexes is the same as the regular language.
This fact is mentioned in~\cite{morihata2012,DBLP:journals/jip/MiyazakiM19,DBLP:journals/jucs/BerglundML21}.
Atomic grouping is also known to be regular in the same sense~\cite{DBLP:conf/wia/BerglundMWW17}. 
However, it is known that look-ahead and atomic grouping can make the number of states of the corresponding DFA grow exponentially~\cite{morihata2012,DBLP:journals/jip/MiyazakiM19,DBLP:journals/jucs/BerglundML21,DBLP:conf/wia/BerglundMWW17}.

In what follows, for simplicity, we only discuss positive look-ahead in discussions of look-around.
Adaptation to other look-around operators, such as negative look-behind, is straightforward.

\subsection{NFAs with Sub-automata}\label{ssec:sub-automata}


We introduce \emph{NFAs with sub-automata} for backtracking matching algorithms for (la, at)-regexes. This extended notion of NFAs is suggested in~\cite[Section~IX.B]{DBLP:conf/sp/DavisSL21}, but it seems ours is the first formal exposition.

Roughly speaking, an NFA with sub-automata is an NFA whose transitions can be labeled with---in addition to a character $\sigma\in \Sigma$, as in usual NFAs---another NFA with sub-automata. See \cref{fig:la-at-conversion}, where transitions from $q_{0}$ to $q_{1}$ are labeled with \fbox{$r$}, the NFA with sub-automata obtained by converting $r$. We annotate these transitions further with a label ($\mathsf{pla}$ for positive look-ahead, $\mathsf{at}$ for atomic grouping, etc.) that indicates which operator they arise from. Note that NFAs with sub-automata can be nested---transitions in  \fbox{$r$} in \cref{fig:la-at-conversion} can be labeled with NFAs with sub-automata, too.

Our precise definition is as follows. There, $P$ is the set that collects all states that occur in an NFA with sub-automata $\mathcal{A}$, i.e., in \begin{enumerate*}[label=\arabic*)]\item the top-level NFA, \item its label NFAs, \item their label NFAs,\end{enumerate*} and so on.

\begin{definition}[NFAs with sub-automata]\label{def:NFAwSubautom}
  An \emph{NFA with sub-automata} $\mathcal{A}$ is a quintuple $\mathcal{A}=(P, Q, q_0, F, T)$ where $P$ is a finite set of states and $Q \subseteq P$ is a set of so-called \emph{top-level states}.
  We require that the quadruple $(Q, q_0, F, T)$ is an NFA, except that the value $T(q)$ of the transition function $T$ is either
  \begin{enumerate*}[label=\arabic*)]
    \item $\mathsf{Eps}(q')$, $\mathsf{Branch}(q',q'')$, or $\mathsf{Char}(\sigma,q')$ (as in usual NFAs,  \cref{ssec:nfa}), or
    \item $\mathsf{Sub}(k, \mathcal{A}', q')$, where  $\mathcal{A}'$ is an NFA with sub-automata, $q'$ is a successor state, and $k$ is a \emph{kind label} where $k\in\{\mathsf{pla}, \mathsf{nla}, \mathsf{plb}, \mathsf{nlb}, \mathsf{at}\}$.
  \end{enumerate*}

  We further impose the following requirements.
  Firstly, we require all NFAs with sub-automata in $\mathcal{A}$ to have disjoint state spaces.
  That is, for any distinct top-level states $q,q''\in Q$ in $\mathcal{A}$, if $T(q) = \mathsf{Sub}(k, \mathcal{A}', q')$ and $T(q'') = \mathsf{Sub}(k', \mathcal{A}'', q''')$, then we must have $P' \cap P'' = \emptyset$, $Q \cap P' = \emptyset$ and  $Q \cap P'' = \emptyset$, where $\mathcal{A}' = (P', \dotsc)$ and $\mathcal{A}'' = (P'', \dotsc)$.
  Secondly, we require that the set $P$ in $\mathcal{A}=(P,\dotsc)$ is the (disjoint) union of all states that occur within $\mathcal{A}$, that is, $P = Q \cup \bigcup_{q\in Q, T(q) = \mathsf{Sub}(k, \mathcal{A}', q'), \mathcal{A}' = (P', \dotsc)} P'$.
\end{definition}

The kind label $k$
in $\mathsf{Sub}(k, \mathcal{A}', q'')$ indicates how the sub-automaton $\mathcal{A}'$ should be used (cf.\ \cref{alg:la-at-matching}). If every kind label occurring in $\mathcal{A}$ (including its sub-automata) is either 
$\mathsf{pla}, \mathsf{nla}, \mathsf{plb}$, or $\mathsf{nlb}$, then $\mathcal{A}$ is called a \emph{la-NFA}. Similarly, if every kind label is $\mathsf{at}$, $\mathcal{A}$ is called an \emph{at-NFA}. Following this convention, general NFAs with sub-automata are called \emph{(la, at)-NFAs}.

Note that the definition is recursive. Non-well-founded nesting is prohibited, however, by the finiteness of $P$. 
By the definition, if $P = Q$, then $\mathcal{A}$ does not contain any transitions labeled with sub-automata. 

In addition to $\mathsf{Eps}$ and $\mathsf{Branch}$ transitions, we refer to $\mathsf{Sub}$ transitions with a label $k \in \{ \mathsf{pla}, \mathsf{nla}, \mathsf{plb}, \mathsf{nlb} \}$ as $\varepsilon$-transitions too.
We also assume the following, similarly to \cref{assum:no-eps-loops}.

\begin{assumption}\label{assum:la-at-no-eps-loops}
  (la, at)-NFAs do not contain $\varepsilon$-loops.
\end{assumption}


For (la, at)-regexes, their conversion to (la, at)-NFAs is described by the constructions in \cref{fig:la-at-conversion}---using transitions labeled with sub-automata---in addition to the conversion for regexes in \cref{ssec:nfa}.
Note that we have $|P| = O(|r|)$ in these constructions.

\begin{figure}[tb]
  \centering
  \begin{minipage}[t]{0.45\textwidth}
    \centering
    \begin{tikzpicture}[initial text=, node distance=2cm, on grid,auto,scale=0.75]
      \scriptsize
      \node[state,initial] (q_0) at (0, 0) {$q_0$};
      \node[state,accepting] (q_1) at (3, 0) {$q_1$};
      \path[->] (q_0) edge node {$\mathsf{Sub}(\mathsf{pla}, \boxed{\rule{0pt}{1.6ex}\hspace*{0.5ex}r\hspace*{0.5ex}})$} (q_1);
    \end{tikzpicture}
    \subcaption{$( \texttt{?=} r )$ (positive look-ahead)}
  \end{minipage}
  \begin{minipage}[t]{0.45\textwidth}
    \centering
    \begin{tikzpicture}[initial text=, node distance=2cm, on grid,auto,scale=0.75]
      \scriptsize
      \node[state,initial] (q_0) at (0, 0) {$q_0$};
      \node[state,accepting] (q_1) at (3, 0) {$q_1$};
      \path[->] (q_0) edge node {$\mathsf{Sub}(\mathsf{at}, \boxed{\rule{0pt}{1.6ex}\hspace*{0.5ex}r\hspace*{0.5ex}})$} (q_1);
    \end{tikzpicture}
    \subcaption{$( \texttt{?>} r )$ (atomic grouping)}
  \end{minipage}
  \caption{a conversion from (la, at)-regexes to (la, at)-NFAs. For negative look-ahead, we use the corresponding kind label $\mathsf{nla}$. For positive and negative look-behind, besides using the kind labels $\mathsf{plb}$ and $\mathsf{nlb}$, we suitably reverse \fbox{$r$}.}\label{fig:la-at-conversion}
\end{figure}
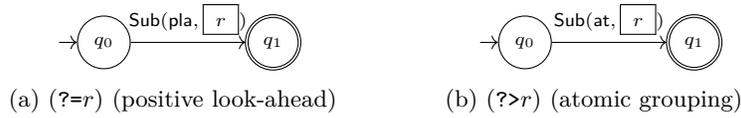

\begin{algorithm}[tb]
  \caption{a partial backtracking matching algorithm for (la, at)-NFAs}\label{alg:la-at-matching}
  \begin{algorithmic}[1]
    \Function{\MatchLaAtAw}{$q, i$}
      \Statex\hspace*{\algorithmicindent}
      \begin{tabular}{rl} 
       \textbf{Parameters}: & a (la, at)-NFA $\mathcal{A}$, and an input string $w$ \\
            \textbf{Input}: & a current state $q$, and a current position $i$ \\
           \textbf{Output}: & returns $\mathsf{SuccessAt}(i')$ if the matching succeeds, or \\
                            & returns $\mathsf{Failure}$ if the matching fails
      \end{tabular}
      \State$(P, Q, q_0, F, T) = \mathcal{A}$
      \If{$q \in F$}
        \Statex\hspace*{\algorithmicindent}$\vdots$ \Comment{the same as \cref{alg:matching}}    
      \setcounter{ALG@line}{11}
      \ElsIf{$T(q) = \mathsf{Sub}(\mathsf{pla}, \mathcal{A}', q')$}\label{line:pla-matching-1}  
     \Statex\Comment{positive look-ahead; other look-around ops.\ are similar}
        \State$(P', Q', q_0', F', T') = \mathcal{A}'$\label{line:pla-matching-2}
        \State$\mathsf{result} \gets \Call{\MatchLaAtApw}{q_0', i}$\label{line:pla-matching-3}
        \If{$\mathsf{result} = \mathsf{SuccessAt}(i')$}\label{line:pla-matching-4}
          \Return$\Call{\MatchLaAtAw}{q', i}$
        \Else\ \label{line:pla-matching-5}
          \Return$r$
        \EndIf%
      \ElsIf{$T(q) = \mathsf{Sub}(\mathsf{at}, \mathcal{A}', q')$}\label{line:at-matching-1}
         \Comment{atomic grouping}
        \State$(P', Q', q_0', F', T') = \mathcal{A}'$\label{line:at-matching-2}
        \State$\mathsf{result} \gets \Call{\MatchLaAtApw}{q_0', i}$\label{line:at-matching-3}
        \If{$\mathsf{result} = \mathsf{SuccessAt}(i')$}\label{line:at-matching-4}
          \Return$\Call{\MatchLaAtAw}{q', i'}$
        \Else\ \label{line:at-matching-5}
          \Return$r$
        \EndIf%
      \EndIf%
    \EndFunction%
  \end{algorithmic}
\end{algorithm}

The backtracking matching algorithm in~\cref{alg:matching} can be naturally extended to (la, at)-NFAs; it is shown in \cref{alg:la-at-matching}.
The clauses for positive look-ahead (\crefrange{line:pla-matching-1}{line:pla-matching-5}) and atomic grouping (\crefrange{line:at-matching-1}{line:at-matching-5}) are similar to each other, conducting matching for sub-automata. Note that their difference is in the ``return position'' ($i$ in \cref{line:pla-matching-4}; $i'$ in \cref{line:at-matching-4}).

The clauses for other look-around operators are similar to the ones for positive look-around. For look-behind, we can suitably use an additional parameter $d \in \{ -1, +1 \}$ for indicating a matching direction.

Using the extended backtracking matching algorithm (\cref{alg:la-at-matching}),
we define the  \emph{partial matching problem for (la, at)-regexes} in the same way as for regexes without extensions (\cref{prob:regexParMatch}).

\begin{problem}[(la, at)-regex partial matching]\label{prob:laAtRegexMatch}\par\noindent
  \begin{tabular}{rl}
    \textbf{Input}: & a (la, at)-regex $r$, an input string $w$, and a starting position $i$ \\
    \textbf{Output}: & returns $\Call{\MatchLaAtArw}{q_0, i}$ where $\mathcal{A}(r) = (P, Q, q_0, F, T)$.
  \end{tabular}
\end{problem}

\section{Previous Works on Regex Matching with Memoization}\label{sec:prev-work}


This section introduces an existing work~\cite{DBLP:conf/sp/DavisSL21} on regex matching with memoization, paving the way for our algorithms for (la, at)-regexes in \cref{sec:memo-la,sec:memo-at}. 




\emph{Memoization} is a programming technique that makes recursive computations more efficient by
\begin{enumerate*}[label=\arabic*)]
  \item recording arguments of a function and the corresponding return values and
  \item reusing them when the function is called with the recorded arguments.
\end{enumerate*}


As we described in \cref{ssec:backtracking}, regex matching is conducted by backtracking matching. 
It is implemented by recursive functions (see \cref{alg:matching,alg:la-at-matching}); thus, it is a natural idea to apply memoization.
Since Java 14, Java's regex implementation has indeed used memoization for optimization.
However, this optimization is not enough to completely prevent ReDoS;\@ see, e.g.,~\cite{DBLP:conf/wia/MerweMLB21}.

The work that inspires the current work the most is~\cite{DBLP:conf/sp/DavisSL21}, whose main novelty is linear-time backtracking regex matching (much like the current work). Its contributions are as follows. 
\begin{enumerate}
  \item\label{davis21Contrib1}
    Focusing on (non-extended) regexes (see~\cref{ssec:regexes}), they introduce a backtracking matching algorithm that uses memoization.
    It achieves a linear-time complexity: for an input string $w$, its runtime is $O(|w|)$. 
  \item\label{davis21Contrib2}
    They introduce \emph{selective memoization}, by which they reduce the domain of the memoization table from $Q\times \mathbb{N}$ to $Q_{\textsf{sel}}\times \mathbb{N}$.
    Here $Q_{\textsf{sel}}$ is a subset of $Q$ that is often much smaller.
  \item
    They introduce a memory-efficient compression method---based on run-length encoding (RLE)---for memoization tables. 
   \item\label{davis21Contrib4}
     Finally, they discuss adaptations of the above method to extended regexes, namely REWZWA (the extension by look-around; look-around is called \emph{zero-width assertion} in~\cite{DBLP:conf/sp/DavisSL21}) and REWBR (the extension by \emph{back-reference}). 
\end{enumerate}
We will mainly discuss the above item~\ref{davis21Contrib1}; it serves as a basis for our algorithms in \cref{sec:memo-la,sec:memo-at}. The technique in the item~\ref{davis21Contrib2} is potentially very relevant: we expect that it can be combined with the current work; doing so is future work. The content of the item~\ref{davis21Contrib2} is reviewed in \FinalOrArxiv{\cite[Appendix~A]{FujinamiH24ESOP_arxiv_extended_ver}}{\cref{appendix:selectiveMemo}} for the record.

\begin{remark}\label{rem:davis2021Err}
  On the above item~\ref{davis21Contrib4}, the work~\cite{DBLP:conf/sp/DavisSL21} claims that the time complexity of their algorithm is linear also for REWZWA ($O(|w|)$ for an input string $w$).
  However, we believe that this claim comes with the following problems.
  \begin{itemize}
    \item
      The description of an algorithm for REZWA in~\cite{DBLP:conf/sp/DavisSL21} is abstract and leaves room for interpretation. 
      The description is to ``preserve the memo functions of the sub-automata throughout the simulation of the top-level M-NFA, remembering the results from sub-simulations that begin at different indices $i$ of $w$'' \cite[Section~IX-B]{DBLP:conf/sp/DavisSL21}. For example, it is not explicit what the ``results'' are---they can mean (complete) matching results or mere success/failure.
    \item
      Moreover, the part ``that begin at different indices $i$ of $w$'' is problematic; we believe that remembering these results  does not lead to linear-time complexity. 
      This point is discussed later in \cref{rem:davis2021-err-la}.
    \item
      Besides, there is a gap between the algorithm described in the paper~\cite{DBLP:conf/sp/DavisSL21} and its prototype implementation~\cite{davis2021impl}, even for (non-extended) regexes. See \cref{rem:memo-timing}.
    \item
       Because of this gap, the implementation~\cite{davis2021impl} works in linear time for all regexes, including REZWA, but can lead to erroneous results for REZWA. See \cref{rem:davis2021-err-la}. 

  \end{itemize}
Our contribution includes a correct memoization algorithm for look-around (REZWA) that resolves the above problems.

\end{remark}


% ...

\subsection{Linear-time Backtracking Matching with Memoization}\label{ssec:davisAlgo}

\begin{algorithm}[tbp]
  \caption{a total matching algorithm with memoization for NFAs without $\varepsilon$-transitions~\cite[Listing~2]{DBLP:conf/sp/DavisSL21}.
  }\label{alg:davis-matching}
  \begin{algorithmic}[1]
    \Function{\DavisSLMAw}{$q, i$}
    \Statex\hspace*{\algorithmicindent}
    \begin{tabular}{rl} 
     \textbf{Parameters}: & an NFA $\mathcal{A}$ without $\varepsilon$-transitions, an input string $w$, and \\
                          & a memoization table $M\colon Q \times \mathbb{N} \rightharpoonup \{ \textbf{false} \}$ \\
          \textbf{Input}: & a current state $q$, and a current position $i$ \\
         \textbf{Output}: & returns \textbf{true} if the matching succeeds, or \\
                          & returns \textbf{false} if the matching fails
    \end{tabular}
      \State$(Q, q_0, F, \delta) = \mathcal{A}$
      \If{$i = |w|$}
        \Return whether $q \in F$
      \EndIf%
      \If{$M(q, i) \ne \bot$}
        \Return$M(q, i)$
      \EndIf%
      \For{$q' \in \delta(q, w[i])$}
        \If{\Call{\DavisSLMAw}{$q', i + 1$}}
          \Return$\textbf{true}$
        \EndIf%
      \EndFor%
      \State$M(q, i) \leftarrow \textbf{false}$ \label{line:algo3MemoiFalse}
      \State\Return$\textbf{false}$
    \EndFunction%
  \end{algorithmic}
\end{algorithm}

We describe the first main contribution of the work~\cite{DBLP:conf/sp/DavisSL21} (the item~\ref{davis21Contrib1} in the above list), namely a  backtracking algorithm that achieves a linear-time complexity thanks to memoization. The algorithm~\cite[Listing~2]{DBLP:conf/sp/DavisSL21} is presented in \cref{alg:davis-matching}.

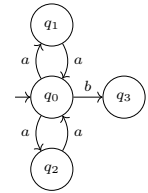
\begin{wrapfigure}[9]{r}[0pt]{3cm}
  \centering
  \vspace{-2.2em}
  \scalebox{.8}{\begin{tikzpicture}[initial text=, node distance=2cm, on grid,auto,scale=0.8]
    \scriptsize
    \node[state,initial] (q_0) at (0, 0) {$q_0$};
    \node[state] (q_1) at (0, 1.5) {$q_1$};
    \node[state] (q_2) at (0, -1.5) {$q_2$};
    \node[state] (q_3) at (1.5, 0) {$q_3$};
    \path[->] (q_0) edge [bend left] node {$a$} (q_1);
    \path[->] (q_1) edge [bend left] node {$a$} (q_0);
    \path[->] (q_0) edge [bend right] node [left] {$a$} (q_2);
    \path[->] (q_2) edge [bend right] node [right] {$a$} (q_0);
    \path[->] (q_0) edge node {$b$} (q_3);
  \end{tikzpicture}}
  \caption{the NFA $\mathcal{A}({(aa|aa)}^\ast b)$, after removing $\varepsilon$-transitions}\label{fig:nfa:aaoraastar}
\end{wrapfigure}
In this algorithm $\textproc{\DavisSLMAw}$, an NFA $\mathcal{A}$ is a quintuple $(Q, q_0, F, \delta)$ where $\delta\colon Q \times \Sigma \to 2^Q$ is an non-deterministic transition function.
An additional parameter $M\colon Q \times \mathbb{N} \rightharpoonup \{ \textbf{false} \}$ is a memoization table, which is mathematically a mutable partial function.
This algorithm implements total matching (cf.\ \cref{ssec:backtracking}).
It is notable that the memoization table records only matching \emph{failures}: a matching \emph{success} does not have to be recorded since it immediately propagates to the success of the whole problem.


This algorithm achieves a linear-time matching.
It thus prevents ReDoS.
A full proof of linear-time complexity is found in~\cite[Appendix~C]{DBLP:conf/sp/DavisSL21}, but its essence is the following (note the critical role of memoization here).
\begin{itemize}
 \item For any call $\Call{\DavisSLMAw}{q, i}$, 
 if $M(q, i)$ is defined, then the call does not invoke any further recursive calls. 
 \item When such a call returns $\mathbf{false}$, the entry $M(q,i)$ of the memoization gets defined (\cref{line:algo3MemoiFalse}).
 \item As a consequence,  the number of recursive calls of  $\textproc{\DavisSLMAw}$ is limited to $|Q|\times|w|$.
\end{itemize}

\begin{example}[matching with memoization for NFAs without $\varepsilon$-transitions]
   Let us consider the regex $(aa|aa)^\ast b$ and the corresponding NFA $\mathcal{A}((aa|aa)^\ast b)$ defined in \cref{ssec:nfa}.
   For the purpose of applying \cref{alg:davis-matching}, we manually remove its $\varepsilon$-transitions, leading to the NFA in \cref{fig:nfa:aaoraastar}. Let $w = \texttt{"$a^{2n} c$"}$ be an input string.
  \Call{\MatchAw}{$q_0, 0$} (without memoization) invokes recursive calls $O(2^n)$ times for the same reason as in \cref{exa:catastrophic}, but \Call{\DavisSLMzAw}{$q_0, 0$} (with memoization, where $M_0$ is the initial memoization table) invokes recursive calls $O(n)$ times because $M(q_0,i)$ for each position $i\in \{ 0, 2, \dots, 2n \}$ has been recorded after the first visit.
\end{example}

\begin{algorithm}[tbp]
  \caption{a variant of \cref{alg:davis-matching} implemented in their prototype~\cite{davis2021impl}}\label{alg:memoization-before-trans}
  \begin{algorithmic}[1]
    \Function{\DavisSLImplMAw}{$q, i$}
      \State$(Q, q_0, F, \delta) = \mathcal{A}$
      \If{$i = |w|$}
        \Return whether $q \in F$
      \EndIf%
      \If{$M(q, i) \ne \bot$}
        \Return$M(q, i)$
      \EndIf%
      \Statex%
            \hspace*{\algorithmicindent}\textcolor{red}{$M(q, i) \gets \mathbf{false}$}
            \Comment{$M(q,i)$ is speculatively set to $\mathbf{false}$}
      \For{$q' \in \delta(q, w[i])$}
        \Statex\hspace*{\algorithmicindent}$\vdots$
      \EndFor%
      \setcounter{ALG@line}{6}
      \State \sout{$M(q, i) \gets \mathbf{false}$}
            \Comment{moved up}
      \Statex$\vdots$
    \EndFunction%
  \end{algorithmic}
\end{algorithm}

\begin{remark}\label{rem:memo-timing}
Following the discussion in  \cref{rem:davis2021Err},
here we describe the gap between \cref{alg:davis-matching}---the algorithm described in the paper~\cite{DBLP:conf/sp/DavisSL21}---and its prototype implementation~\cite{davis2021impl}. The latter is shown in \cref{alg:memoization-before-trans}. 

The precise difference between the two algorithms is that \cref{line:algo3MemoiFalse} in \cref{alg:davis-matching} is moved up to the moment just before the for-loop, in \cref{alg:memoization-before-trans}. It is not hard to see that this modification does not affect the correctness of the algorithm: if the pair $(q, i)$ is visited again in the future, it means that the current matching from $(q, i)$ did not succeed, and backtracking occurred. Note that, in case the current matching is successful, the function call returns $\mathbf{true}$ so the memoization content $M(q, i)$ should not matter.

However, the above argument is true only when there is no look-around. (A detailed discussion is in \cref{exa:la-issue}.)
This point seems to be missed in the implementation~\cite{davis2021impl}.
\end{remark}

\subsection{Matching with Memoization Adapted to the Current Formalism}\label{ssec:memo}

In \cref{alg:memoization}, we present an adaptation of \cref{alg:davis-matching} to our formalism, especially our definition of NFA (\cref{ssec:nfa}) that offers fine-grained handling of nondeterminism. \cref{alg:memoization} has been adapted also to solve \emph{partial} matching (it returns a matching position $i'$) rather than total matching as in \cref{alg:davis-matching} (cf.\ \cref{ssec:backtracking}). \cref{alg:memoization} serves as a basis towards our extensions to look-around and atomic grouping in \crefrange{sec:memo-la}{sec:memo-at}.

The adaptation is straightforward: \cref{line:memoization-trans-begin} ensures that the algorithm solves partial matching; the rest is a natural adaptation of the for-loop of \cref{alg:davis-matching} to our definition of NFA (\cref{ssec:nfa}). The algorithm terminates thanks to \cref{assum:no-eps-loops}. 
We note that the type of memoization tables does not have to be changed compared to \cref{alg:davis-matching}.


\begin{algorithm}[tbp]
  \caption{a partial matching algorithm with memoization.  An adaptation of \cref{alg:davis-matching} from~\cite{DBLP:conf/sp/DavisSL21}, and a basis for our algorithms (\cref{alg:la-memoization,alg:at-memoization})}\label{alg:memoization}
  \begin{algorithmic}[1]
    \Function{\MemoMAw}{$q, i$}
      \Statex\hspace*{\algorithmicindent}
      \begin{tabular}{rl} 
       \textbf{Parameters}: & an NFA $\mathcal{A}$, an input string $w$, and \\
                            & a memoization table $M\colon Q \times \mathbb{N} \rightharpoonup \{ \mathsf{Failure} \}$ \\
            \textbf{Input}: & a current state $q$, and a current position $i$ \\
           \textbf{Output}: & returns $\mathsf{SuccessAt}(i')$ if the matching succeeds, or \\
                            & returns $\mathsf{Failure}$ if the matching fails
      \end{tabular}
      \State$(Q, q_0, F, T) = \mathcal{A}$
      \If{$M(q, i) \ne \bot$}\label{line:memoization-ret}
        \Return$M(q, i)$
      \EndIf%
      \State$\mathsf{result} \gets \bot$
      \If{$q \in F$}\label{line:memoization-trans-begin}
          $\mathsf{result} \gets \mathsf{SuccessAt}(i)$
      \ElsIf{$T(q) = \mathsf{Eps}({q'})$}
          $\mathsf{result} \gets \Call{\MemoMAw}{q', i}$
      \ElsIf{$T(q) = \mathsf{Branch}({q'}, q'')$}
        \State$\mathsf{result} \leftarrow \Call{\MemoMAw}{q', i}$
        \If{$\mathsf{result} = \mathsf{Failure}$}
            $\mathsf{result} \gets \Call{\MemoMAw}{q'', i}$
        \EndIf%
      \ElsIf{$T(q) = \mathsf{Char}(\sigma, q')$}
        \If{$i < |w|$ \textbf{and} $w[i] = \sigma$}
            $\mathsf{result} \gets \Call{\MemoMAw}{q', i + 1}$
        \Else\ 
            $\mathsf{result} \gets \mathsf{Failure}$
        \EndIf%
      \EndIf%
      \label{line:memoization-trans-end}
      \If{$\mathsf{result} = \mathsf{Failure}$}\label{line:memoization-record}
        $M(q, i) \gets \mathsf{Failure}$
      \EndIf%
      \State\Return$\mathsf{result}$
      \Comment{$\mathsf{result}\neq\bot$, as one can easily see}
    \EndFunction%
  \end{algorithmic}
\end{algorithm}

\cref{alg:memoization} exhibits the same desired properties as \cref{alg:davis-matching}, namely correctness (with respect to \cref{prob:regexParMatch}) and linear-time complexity. We formally state these properties for the record; here,  $M_0\colon Q \times \mathbb{N} \rightharpoonup \{ \mathsf{Failure} \}$ is the initial memoization table (its entry is anywhere $\bot$). 

\noindent
\begin{minipage}{\textwidth}
 \begin{theorem}[linear-time complexity of \cref{alg:memoization}]\label{thm:algo-lin}
  For an NFA $\mathcal{A} = (Q, q_0, F, T)$, an input string $w$, and an position $i \in \{ 0, \dots, |w| \}$, $\Call{\MemoMzAw}{q_0, i}$ terminates with $O(|w|)$ recursive calls.
 \end{theorem}
\end{minipage}

\noindent
\begin{minipage}{\textwidth}
\begin{theorem}[correctness of \cref{alg:memoization}]\label{thm:algo-cor}
  For an NFA $\mathcal{A} = (Q, q_0, F, T)$, an input string $w$, and an position $i \in \{ 0, \dots, |w| \}$, $\Call{\MatchAw}{q_0, i} = \Call{\MemoMzAw}{q_0, i}$.
\end{theorem}
\end{minipage}

\vspace{.3em}
The proofs can be found \FinalOrArxiv{in~\cite[Appendix~B.1]{FujinamiH24ESOP_arxiv_extended_ver}}{in \cref{appendix:memo-proof}}. Here is their outline.

We first introduce the notion of run for $\textproc{Match}$ and $\textproc{Memo}$; it records \emph{recursive calls} of the function itself, as well as \emph{invocations} of the memoization table, together with their return values.

For linear time complexity (\cref{thm:algo-lin}), we show that \begin{enumerate*}[label=\arabic*)]\item a recursive call with the same argument $(q,i)$ appears at most once in a run, and that \item the number of invocations of the memoization table with the same key $(q,i)$ is bounded by the (graph-theoretic) in-degree.\end{enumerate*}
Linear-time complexity then follows easily.

For correctness (\cref{thm:algo-cor}), we introduce a conversion from runs of $\textproc{Memo}$ to runs of $\textproc{Match}$. By showing that \begin{enumerate*}[label=\arabic*)]\item the result is indeed a valid run of $\textproc{Match}$ and \item the conversion preserves return values\end{enumerate*}, we show the coincidence of the return values of the two algorithms, i.e., correctness.

\section{Memoization for Regexes with Look-around}\label{sec:memo-la}

We describe our first main technical contribution, namely a backtracking matching algorithm for la-NFAs with memoization (\cref{alg:la-memoization}). We prove that it is correct (\cref{thm:la-cor}) and that its time complexity is linear ($O(|w|)$, \cref{thm:la-lin}).  

The key ingredient of our algorithm is the type of memoization tables, where their range is extended from $\{\mathsf{Failure}\}$ to $\{\mathsf{Failure}, \mathsf{Success}\}$.
We motivate this extension through two problematic algorithms $\textproc{\MemoExitLa}$ and $\textproc{\MemoEnterLa}$;
$\textproc{\MemoExitLa}$ is obtained by naively extending \cref{alg:memoization} ($\textproc{Memo}$) with adding the processing of sub-automaton transitions with $\mathsf{pla}$ (positive look-ahead) done in \cref{alg:la-at-matching} (\crefrange{line:pla-matching-1}{line:pla-matching-5}), and $\textproc{\MemoEnterLa}$ is similar to $\textproc{\MemoExitLa}$, but this records to the memoization table at the same timing as \cref{alg:memoization-before-trans} ($\textproc{DavisSLImpl}$).
In particular, their memoization tables only record $\mathsf{false}$.

The example below shows the problems with the two naive algorithms. Specifically, $\textproc{\MemoExitLa}$ is not linear and $\textproc{\MemoEnterLa}$ is not correct.

\begin{example}\label{exa:la-issue}
  \begin{figure}[tbp]
    \centering
\scalebox{.7}{    \begin{tikzpicture}[initial text=, node distance=2cm, on grid,auto,scale=0.6]
      \tiny
      \node[state,initial] (q_0) at (0, 0) {$q_0$};
      \node[state] (q_1) at (2.5, 1) {$q_1$};
      \node[state] (q_2) at (5.5, 1) {$q_2$};
      \node[state] (q_3) at (7.5, 1) {$q_3$};
      \node[state,accepting] (q_4) at (2.5, -1) {$q_4$};
      \path[-] (q_0) edge node {$\mathsf{Branch}$} (1.5, 0);
      \path[->] (1.5, 0) edge (q_1);
      \path[->] (1.5, 0) edge (q_4);
      \path[->] (q_1) edge node {$\mathsf{Sub}(\mathsf{pla}, \mathcal{A}')$} (q_2);
      \path[->] (q_2) edge node {$\mathsf{Char}(a)$} (q_3);
      \path[-] (q_3) edge (8.5, 1);
      \path[->] (8.5, 1) edge [bend right=2cm] node {$\mathsf{Eps}$} (q_0);
      \begin{scope}[shift={(+7.75, +2.5)}]
        \node at (3.5, 0) {$\mathcal{A}'$};
        \draw (3.1, 0.4) rectangle (10.5, -4.2);
        \node[state,initial] (q_5) at (4.5, -2.5) {$q_5$};
        \node[state] (q_6) at (7, -1.5) {$q_6$};
        \node[state] (q_7) at (9, -1.5) {$q_7$};
        \node[state,accepting] (q_8) at (7, -3.5) {$q_8$};
        \path[-] (q_5) edge node {$\mathsf{Branch}$} (6, -2.5);
        \path[->] (6, -2.5) edge (q_6);
        \path[->] (6, -2.5) edge (q_8);
        \path[->] (q_6) edge node {$\mathsf{Char}(a)$} (q_7);
        \path[-] (q_7) edge (10, -1.5);
        \path[->] (10, -1.5) edge [bend right=2cm] node {$\mathsf{Eps}$} (q_5);
      \end{scope}
    \end{tikzpicture}
}    \caption{the la-NFA $\mathcal{A}(((\texttt{?=} a^\ast) a)^\ast)$}\label{fig:la-example}

\scalebox{.7}{    \begin{tikzpicture}[initial text=, node distance=2cm, on grid,auto,scale=0.6]
      \tiny
      \node[state,initial] (q_0) at (0, 0) {$q_0$};
      \node[state] (q_1) at (2.5, 1) {$q_1$};
      \node[state] (q_2) at (4.5, 1) {$q_2$};
      \node[state] (q_3) at (2.5, -1) {$q_3$};
      \node[state] (q_4) at (5.5, -1) {$q_4$};
      \node[state] (q_5) at (7.5, -1) {$q_5$};
      \node[state,accepting] (q_6) at (9.5, -1) {$q_6$};
      \path[-] (q_0) edge node {$\mathsf{Branch}$} (1.5, 0);
      \path[->] (1.5, 0) edge (q_1);
      \path[->] (1.5, 0) edge (q_3);
      \path[->] (q_1) edge node {$\mathsf{Char}(a)$} (q_2);
      \path[-] (q_2) edge (5.5, 1);
      \path[->] (5.5, 1) edge [bend right=2cm] node {$\mathsf{Eps}$} (q_0);
      \path[->] (q_3) edge node {$\mathsf{Sub}(\mathsf{at}, \mathcal{A
      }')$} (q_4);
      \path[->] (q_4) edge node {$\mathsf{Char}(a)$} (q_5);
      \path[->] (q_5) edge node {$\mathsf{Char}(b)$} (q_6);
      \begin{scope}[shift={(+7.75, +2.5)}]
        \node at (3.5, 0) {$\mathcal{A}'$};
        \draw (3.1, 0.4) rectangle (10.5, -4.2);
        \node[state,initial] (q_7) at (4.5, -2.5) {$q_7$};
        \node[state] (q_8) at (7, -1.5) {$q_8$};
        \node[state] (q_9) at (9, -1.5) {$q_9$};
        \node[state,accepting] (q_10) at (7, -3.5) {$q_{10}$};
        \path[-] (q_7) edge node {$\mathsf{Branch}$} (6, -2.5);
        \path[->] (6, -2.5) edge (q_8);
        \path[->] (6, -2.5) edge (q_10);
        \path[->] (q_8) edge node {$\mathsf{Char}(a)$} (q_9);
        \path[-] (q_9) edge (10, -1.5);
        \path[->] (10, -1.5) edge [bend right=2cm] node {$\mathsf{Eps}$} (q_7);
      \end{scope}
    \end{tikzpicture}}
    \caption{the at-NFA $\mathcal{A}(a^\ast (\texttt{?>} a^\ast) ab)$}\label{fig:at-example-1}
  \end{figure}

 Consider the la-NFA $\mathcal{A} = \mathcal{A}(((\texttt{?=} a^\ast) a)^\ast) = (P, Q, q_0, F, T)$ shown in \cref{fig:la-example}, and let $w = \texttt{"$a^n$"}$ be an input string. 

 $\Call{\MemoExitLaMzAw}{q_0, 0}$ invokes recursive calls $O(|w|^2)$ times---in the same way as $\textproc{\MatchLaAt}$---because there are no matching failures in $\mathcal{A}'$ that contribute to memoization.

We also see $\textproc{\MemoEnterLa}$ is not correct: $\Call{\MatchLaAtAw}{q_0 ,0}$ returns $\mathsf{SuccessAt}(n)$, but $\Call{\MemoEnterLaMzAw}{q_0, 0}$ returns $\mathsf{SuccessAt}(1)$ because $M(q_5, 1)=\mathbf{false}$ is recorded during the first loop and interpreted as a matching failure.
\end{example}

In \cref{exa:la-issue}, a natural solution to the non-linearity issues with $\textproc{\MemoExitLa}$ is to enrich memoization so that it also records previous successes of look-around.
Furthermore, since matching positions do not matter in look-around, the type of memoization tables should be $M\colon P \times \mathbb{N} \rightharpoonup \{ \mathsf{Failure}, \mathsf{Success} \}$.

\begin{remark}\label{rem:davis2021-err-la}
  The work~\cite[Section~IX-B]{DBLP:conf/sp/DavisSL21} proposes an adaptation of their memoization algorithm to REZWA. Its description in~\cite[Section~IX-B]{DBLP:conf/sp/DavisSL21}  (to ``preserve the memo functions$\dotsc$''; see \cref{rem:davis2021Err}) consists of the following two points:
  \begin{enumerate}
    \item preserving the memoization tables of the sub-automata throughout the whole matching, and
    \item recording the results of sub-automata matching from different start positions $i$ of $w$.
  \end{enumerate}
  The naive algorithm 
  $\textproc{\MemoExitLa}$ we discussed above implements the first point. We can further add the second point (that is essentially ``memoization for sub-automaton matching'') to $\textproc{\MemoExitLa}$. 

  However, we find that this is not enough to ensure linear-time complexity. The problem is that the ``memoization for sub-automaton matching'' is used too infrequently. For example, in \cref{exa:la-issue}, the start positions of sub-automaton matching are different each time; thus, the memoized results are never used.

  Our algorithm (\cref{alg:la-memoization}) resolves this problem by letting the memoization tables (for sub-automaton matching) record results not only for \emph{starting} positions but also for non-starting positions. 

  We also note that there is a gap between the algorithm in the paper~\cite{DBLP:conf/sp/DavisSL21} and its prototype implementation~\cite{davis2021impl}; see \cref{rem:memo-timing}. The latter is linear time but not always correct. For example, in \cref{exa:la-issue}, the correct result is $\mathsf{SuccessAt}(n)$, but the prototype~\cite{davis2021impl} returns $\mathsf{SuccessAt}(1)$, similarly to $\textproc{\MemoEnterLa}$.
\end{remark}

\begin{algorithm}[tbp]
  \caption{our partial matching algorithm with memoization for la-NFAs}\label{alg:la-memoization}
  \begin{algorithmic}[1]
    \Function{\MemoLaMAw}{$q, i$}
      \Statex\hspace*{\algorithmicindent}
      \begin{tabular}{rl} 
       \textbf{Parameters}: & a la-NFA $\mathcal{A}$, an input string $w$, and \\
                            & a memoization table $M\colon P \times \mathbb{N} \rightharpoonup \{ \mathsf{Failure}, \mathsf{Success} \}$ \\
            \textbf{Input}: & a current state $q$, and a current position $i$ \\
           \textbf{Output}: & returns $\mathsf{SuccessAt}(i')$ if the matching succeeds, \\
                            & returns $\mathsf{Success}$ if a matching success is in $M$ (cf.\ \cref{lem:success}), or \\
                            & returns $\mathsf{Failure}$ if the matching fails
      \end{tabular}
      \State$(P, Q, q_0, F, T) = \mathcal{A}$
      \If{$M(q, i) \ne \bot$}
        \Return$M(q, i)$
      \EndIf%
      \State$\mathsf{result} \gets \bot$
      \If{$q \in F$} \Comment{the same as \crefrange{line:memoization-trans-begin}{line:memoization-trans-end} of \cref{alg:memoization}}      
      \Statex\hspace*{\algorithmicindent}$\vdots$  \setcounter{ALG@line}{12}
      \ElsIf{$T(q) = \mathsf{Sub}(\mathsf{pla}, \mathcal{A}', q')$}\label{line:memo-la-pla-start}
        \State$(P', Q', {q_0}', F', T') = \mathcal{A}'$
        \State$\mathsf{result} \gets \Call{\MemoLaMApw}{{q_0}', i}$
        \If{$\mathsf{result} = \mathsf{SuccessAt}(i')$ \textbf{or} $\mathsf{Success}$}
          \State$\mathsf{result} \gets \Call{\MemoLaMAw}{q', i}$\label{line:memo-la-pla-end}
        \EndIf%
      \EndIf%
      \If{$\mathsf{result} = \mathsf{SuccessAt}(i')$ \textbf{or} $\mathsf{Success}$}
        $M(q, i) \gets \mathsf{Success}$\label{line:memo-la-success}
      \ElsIf{$\mathsf{result} = \mathsf{Failure}$}\ 
        $M(q, i) \gets \mathsf{Failure}$
      \EndIf%
      \State\Return$\mathsf{result}$
    \EndFunction%
  \end{algorithmic}
\end{algorithm}


\cref{alg:la-memoization} is the matching algorithm for la-NFAs that we propose. It adopts the above extended type of $M$.
In \cref{line:memo-la-success}, $\mathsf{Success}$ is recorded in the memoization table when the matching succeeded.
This function can return one of $\mathsf{SuccessAt}(i')$, $\mathsf{Failure}$, and $\mathsf{Success}$.
We first prove the following lemma to see that the algorithm indeed solves the partial matching problem (\cref{prob:laAtRegexMatch}).


\vspace{.3em}
\noindent
\begin{minipage}{\textwidth}
\begin{lemma}\label{lem:success}
  For a la-NFA $\mathcal{A} = (P, Q, q_0, F, T)$, an input string $w$, and a position $i \in \{ 0, \dots, |w| \}$, $\Call{\MemoLaMzAw}{q_0, i}$ returns either $\mathsf{SuccessAt}(i')$ for $i' \in \{ 0, \dots, |w| \}$ or $\mathsf{Failure}$ (it does not return $\mathsf{Success}$).
\end{lemma}
\end{minipage}
\begin{minipage}{\textwidth}
\begin{proof}
When we obtain $\mathsf{Success}$ as a return value, it must be via an entry $M(q, i)$ of the memoization table. However, due to \cref{assum:la-at-no-eps-loops}, when $M(q, i)$ is set to $\mathsf{Success}$ for a state $q$ of the top-level automaton of $\mathcal{A}$, the matching is already finished and returns $\mathsf{SuccessAt}(i')$.
\qed%
\end{proof}
\end{minipage}

\vspace{.3em}
As a consequence of the lemma, we can further shrink the memoization tables in \cref{alg:la-memoization} by not recording $\mathsf{Success}$ for $M(q,i)$ where $q$ is a state of the top-level automaton.



\cref{alg:la-memoization} exhibits the desired properties, namely correctness (with respect to \cref{prob:laAtRegexMatch}) and linear-time complexity.

\vspace{.3em}
\noindent
\begin{minipage}{\textwidth}
\begin{theorem}[linear-time complexity of \cref{alg:la-memoization}]\label{thm:la-lin}
  For a la-NFA $\mathcal{A} = (P, Q, q_0, F, T)$, an input string $w$, and a position $i \in \{ 0, \dots, |w| \}$, $\Call{\MemoLaMzAw}{q_0, i}$ terminates with $O(|w|)$ recursive calls.
\end{theorem}
\end{minipage}
\begin{minipage}{\textwidth}
\begin{theorem}[correctness of \cref{alg:la-memoization}]\label{thm:la-cor}
  For a la-NFA $\mathcal{A} = (P, Q, q_0, F, T)$, an input string $w$, and a position $i \in \{ 0, \dots, |w| \}$, $\Call{\MatchLaAt}{q_0, i} = \Call{\MemoLaMzAw}{q_0, i}$.
\end{theorem}
\end{minipage}

\noindent
 \cref{thm:la-lin,thm:la-cor} can be shown similarly to
 \cref{thm:algo-lin,thm:algo-cor}; \FinalOrArxiv{see~\cite[Appendix~B.2]{FujinamiH24ESOP_arxiv_extended_ver}}{see \cref{appendix:memo-la-proof}}. The in-degree for sub-automata requires some additional care.

\section{Memoization for Regexes with Atomic Grouping}\label{sec:memo-at}

We describe our second main technical contribution, namely a backtracking matching algorithm for at-NFAs with memoization (\cref{alg:at-memoization}).
We prove that it is correct (\cref{thm:at-cor}) and that its time complexity is linear ($O(|w|)$, \cref{thm:at-lin}).

The key ingredient of our algorithm is the type of memoization tables, where their range is extended from $\{\mathsf{Failure}\}$ to $\{ \mathsf{Failure}(j)\mid j \in \{ 0, \dots, \nu(\mathcal{A}_0) \} \}$; the latter records a \emph{depth} $j$ of atomic grouping in order to distinguish failures of different depths.
We motivate this extension through two problematic algorithms $\textproc{\MemoExitAt}$ and $\textproc{\MemoEnterAt}$. Much like in \cref{sec:memo-la}, $\textproc{\MemoExitAt}$  naively extends \cref{alg:memoization} ($\textproc{Memo}$) by adding the processing of sub-automaton transitions with $\mathsf{at}$ done in \cref{alg:la-at-matching} (\crefrange{line:at-matching-1}{line:at-matching-5}), and $\textproc{\MemoEnterLa}$ is similar to $\textproc{\MemoExitAt}$, but records to the memoization table at the same timing as \cref{alg:memoization-before-trans} ($\textproc{DavisSLImpl}$).

Firstly, we observe that $\textproc{\MemoExitAt}$ is not linear for a reason similar to \cref{exa:la-issue}.
(A concrete example is given by \cref{exa:at-issue-1}.)
Therefore, we turn to the other candidate, namely $\textproc{\MemoEnterAt}$.

We find, however, that $\textproc{\MemoEnterAt}$ is also problematic. It is not correct.

\begin{example}\label{exa:at-issue-1}
  \begin{figure}[tbp]
    \centering
  \end{figure}
  
  Consider the at-NFA $\mathcal{A} = \mathcal{A}(a^\ast (\texttt{?>} a^\ast) ab) = (P, Q, q_0, F, T)$ shown in \cref{fig:at-example-1}, and let $w = \texttt{"$a^n b$"}$ be an input string.
  $\Call{\MatchLaAtAw}{q_0, 0}$ returns $\mathsf{Failure}$---the atomic grouping $(\texttt{?>} a^\ast)$ consumes all  $a$'s in $w$ and no $a$ is left for the final $ab$ pattern---but $\Call{\MemoEnterAtMzAw}{q_0, 0}$ returns $\mathsf{SuccessAt}(n+1)$. 
Thus $\textproc{\MemoEnterAt}$ is wrong.

For both algorithms,
  the state $q_7$ in the $\mathsf{at}$ transition is first reached at position $i=n$, and then backtracking is conducted, leading to the state
  $q_7$ again at $i=n-1$. 
 The execution of $\textproc{\MemoEnterAt}$ proceeds as follows.
\begin{itemize}
 \item The first execution path consumes all $a$'s in the loop from $q_0$ to $q_2$, reaches $q_7$ with $i=n$, eventually leading to failure at $q_4$ and thus to backtracking. Speculative memoization ($M(q, i) \gets \mathbf{false}$ in \cref{alg:memoization-before-trans}) is conducted in its course; in particular,  $M(q_7, n) = \mathbf{false}$ is recorded.
 \item After backtracking, the second execution path reaches $q_7$ with $i=n-1$; it then visits $q_8$ once and reaches $q_7$ with $i=n$. Now it uses the memoized value $M(q_7, n) = \mathbf{false}$ (cf.\ Line~4 of \cref{alg:memoization-before-trans}), leading to backtracking to $q_7$ with $i=n-1$. It then takes the branch to $q_{10}$, and the matching for $\mathcal{A}'$ succeeds. Therefore, the execution reaches $q_4$ with $i=n-1$, and the whole matching succeeds.
\end{itemize}
\end{example}

The last example shows the challenge we are facing, namely the need of \emph{distinguishing failures of different depths}. Specifically, in the previous example, the memoized value $M(q_7, n) = \mathbf{false}$ comes from the failure of matching for ambient $\mathcal{A}$; still, it is used to control backtracking in the sub-automaton $\mathcal{A}'$. This fact is problematic in an atomic grouping where, roughly speaking, backtracking in an ambient automaton should not cause backtracking in a sub-automaton. Atomic grouping can be nested, so we must track at which depth failure has happened.

\begin{definition}[nesting depth of atomic grouping]\label{def:nesting-depth}
For an at-NFA $\mathcal{A} = (P, Q, q_0, F, T)$ and a state $q \in P$, the \emph{nesting depth of atomic grouping for $q$}, denoted by $\nu_\mathcal{A}(q)$, is
\begin{equation*}\small
\begin{array}{l}
   \nu_\mathcal{A}(q) = \begin{cases}
    0 & \text{if $q \in Q$} \\
    1 + \nu_{\mathcal{A}'}(q) & \text{where $\mathcal{A}' = (P', Q', q_0', F', T')$} \\
    & \text{s.t. $T(q') = \mathsf{Sub}(\mathsf{at}, \mathcal{A}', q'')$ and $q \in P'$.}
  \end{cases}
\end{array}
\end{equation*}
We also define the \emph{maximum nesting depth of atomic grouping for $\mathcal{A}$}, denoted by $\nu(\mathcal{A})$, as $\nu(\mathcal{A}) = \max_{q \in P} \nu_\mathcal{A}(q)$.
\end{definition}

\begin{algorithm}[tb]
  \caption{our partial matching algorithm with memoization for at-NFAs}\label{alg:at-memoization}
  \begin{algorithmic}[1]
    \Function{\MemoAtMAzAw}{$q, i$}
      \Statex\hspace*{\algorithmicindent}
      \begin{tabular}{rl} 
       \textbf{Parameters}: & an at-NFA $\mathcal{A}_0$, a sub-automaton $\mathcal{A}$ of $\mathcal{A}_0$ (it can be $\mathcal{A}_0$ itself), \\
       	                    & an input string $w$, and a memoization table $M\colon P \times \mathbb{N} \rightharpoonup$ \\
                            & \quad\qquad $\{ \mathsf{Failure}(j)\mid j \in \{ 0, \dots, \nu(\mathcal{A}_0) \} \}$ \\
            \textbf{Input}: & a current state $q$, and a current position $i$ \\
           \textbf{Output}: & returns $\mathsf{SuccessAt}(i', K)$ if the matching succeeds, or \\
                            & returns $\mathsf{Failure}(j)$ if the matching fails
      \end{tabular}
      \State$(P, Q, q_0, F, T) = \mathcal{A}$
      \If{$M(q, i) \ne \bot$}
        \Return$M(q, i)$
      \EndIf%
      \State$\mathsf{result} \gets \bot$
      \If{$q \in F$}
        $\mathsf{result} \gets \mathsf{Success}(i, \emptyset)$
      \ElsIf{$T(q) = \mathsf{Eps}(q')$}
        $\mathsf{result} \gets \Call{\MemoAtMAzAw}{q', i}$
      \ElsIf{$T(q) = \mathsf{Branch}(q', q'')$}
        \State$\mathsf{result} \gets \Call{\MemoAtMAzAw}{q', i}$
        \If{$\mathsf{result} = \mathsf{Failure}(j)$ \textbf{and} $j = \nu_{\mathcal{A}_0}(q)$}\label{line:at-memoization-branch}
          \State$\mathsf{result} \gets \Call{\MemoAtMAzAw}{q'', i}$
          \If{$\mathsf{result} = \mathsf{Failure}(j')$}
            $\mathsf{result} \gets \mathsf{Failure}(\min(j, j'))$
          \EndIf%
        \EndIf%
      \ElsIf{$T(q) = \mathsf{Char}(\sigma, q')$}
        \If{$i < |w|$ \textbf{and} $w[i] = \sigma$}
          $\mathsf{result} \gets \Call{\MemoAtMAzAw}{q', i + 1}$
        \Else\ 
          $\mathsf{result} \gets \mathsf{Failure}(\nu_{\mathcal{A}_0}(q))$
        \EndIf%
      \ElsIf{$T(q) = \mathsf{Sub}(\mathsf{at}, \mathcal{A}', q')$}
        \State$(P', Q', {q_0}', F', T') = \mathcal{A}'$
        \State$\mathsf{result} \gets \Call{\MemoAtMAzApw}{q_0', i}$\label{line:at-memoization-call-1}
        \If{$\mathsf{result} = \mathsf{SuccessAt}(i', K')$}
          \State$\mathsf{result} \gets \Call{\MemoAtMAzAw}{q', i'}$\label{line:at-memoization-call-2}
          \If{$\mathsf{result} = \mathsf{SuccessAt}(i'', K'')$}
            $\mathsf{result} \gets \mathsf{SuccessAt}(i'', K' \cup K'')$
          \ElsIf{$\mathsf{result} = \mathsf{Failure}(j)$}
            \For{$k \in K'$}\label{line:at-memoization-batch}
              $M(k) \gets \mathsf{Failure}(j)$
            \EndFor%
          \EndIf%
        \ElsIf{$\mathsf{result} = \mathsf{Failure}(j)$ \textbf{and} $j > \nu_{\mathcal{A}_0}(q)$}
          $\mathsf{result} \gets \mathsf{Failure}(\nu_{\mathcal{A}_0}(q))$\label{line:at-memoization-reset}
        \EndIf%
      \EndIf%
      \If{$\mathsf{result} = \mathsf{SuccessAt}(i', K)$}
        $\mathsf{result} \gets \mathsf{SuccessAt}(i', K \cup \{ (q, i) \})$
      \ElsIf{$\mathsf{result} = \mathsf{Failure}(j)$}\ 
        $M(q, i) \gets \mathsf{Failure}(j)$\label{line:at-memoization-update}
      \EndIf%
      \State\Return$\mathsf{result}$
    \EndFunction%
  \end{algorithmic}
\end{algorithm}

\cref{alg:at-memoization} is our  algorithm for at-NFAs; the type of its memoization tables is $M\colon P \times \mathbb{N} \rightharpoonup \{ \mathsf{Failure}(j)\mid j \in \{ 0, \dots, \nu(\mathcal{A}) \} \}$.
Some remarks are in order.

Note first that the algorithm takes, as its parameters, the whole at-NFA $\mathcal{A}_{0}$ and its sub-automaton $\mathcal{A}$ as the algorithm's current scope. The top-level call is made with $\mathcal{A}_{0}=\mathcal{A}$  (cf.\ \cref{thm:at-cor,thm:at-lin}); when an $\mathsf{at}$ transition is encountered, the scope goes to the corresponding sub-automaton ($\mathcal{A}'$ in Line~17).

In \cref{line:at-memoization-branch}, the $\mathbf{if}$ condition checks that the nesting depth of $\mathsf{Failure}$ is the depth of the current NFA, and backtracking is performed if and only if it is true. This approach is crucial for avoiding the error in \cref{exa:at-issue-1}.
The rest of the cases for $\mathsf{Eps}, \mathsf{Branch},\mathsf{Char}$ is similar to \cref{alg:memoization}.

The case for $\mathsf{Sub}$ (Lines~15--23) requires some explanation. It is an adaptation of Lines~17--21 of \cref{alg:memoization} with memoization. The apparent complication comes from the set $K$ in $\mathsf{SuccessAt}(i',K)$. The set $K$ is a set of \emph{keys} for a memoization table $M$, that is, pairs $(q,i)$ of a state and a position. The role of $K$ is to collect the set of keys of $M$ for which, once failure happens, the entry $\mathsf{Failure}(j)$ has to be recorded (this is done in a batch manner in Line~22). More specifically, once failure happens in an outer automaton (i.e., at a smaller depth $j$), this has to be recorded as $\mathsf{Failure}(j)$ for inner automata (at greater depths). The set $K$ collects those keys for which this has to be done, starting from inner automata ($\mathcal{A}'$, Line~18) and going to outer ones ($\mathcal{A}$, Lines~19--20). 

A closer inspection reveals that Line~20 is vacuous in \cref{alg:at-memoization}; however, it is needed when we combine it with look-around at the end of the section.




\cref{alg:at-memoization} exhibits the desired properties, namely correctness (with respect to \cref{prob:laAtRegexMatch}) and linear-time complexity.
In \cref{thm:at-cor}, $f$ is a function that converts results of \cref{alg:at-memoization} to results of \cref{alg:la-at-matching}; it is defined by $f(\mathsf{Failure}(j)) = \mathsf{Failure}$ and $f(\mathsf{SuccessAt}(i', K)) = \mathsf{SuccessAt}(i')$.

\vspace{.3em}
\noindent
\begin{minipage}{\textwidth}
\begin{theorem}[linear-time complexity of \cref{alg:at-memoization}]\label{thm:at-lin}
  For an at-NFA $\mathcal{A} = (P, Q, q_0, F, T)$, an input string $w$, and an position $i \in \{ 0, \dots, |w| \}$, $\Call{\MemoAtMzAAw}{q_0, i}$ terminates with $O(|w|)$ recursive calls.
\end{theorem}
\end{minipage}
\begin{minipage}{\textwidth}
\begin{theorem}[correctness of \cref{alg:at-memoization}]\label{thm:at-cor}
  For an at-NFA $\mathcal{A} = (P, Q, q_0, F, T)$, an input string $w$, and an position $i \in \{ 0, \dots, |w| \}$, $\Call{\MatchLaAt}{q_0, i} = f(\Call{\MemoAtMzAAw}{q_0, i})$.
\end{theorem}
\end{minipage}

\vspace{.3em}
\noindent
 \cref{thm:at-lin,thm:at-cor} are proved similarly to \cref{thm:algo-lin,thm:algo-cor}; \FinalOrArxiv{see~\cite[Appendix~B.3]{FujinamiH24ESOP_arxiv_extended_ver}}{see \cref{appendix:memo-at-proof}}. The following points require some extra care. 

Firstly, for linear-time complexity (\cref{thm:at-lin}), there is another recursive call (Line~19) before the return value of a recursive call (Line~17) is memoized (Line~22). If the second recursive call (Line~19) eventually leads to (the same call as) the first call (Line~17) (let's call this event $(\ast)$), then this can nullify the effect of memoization. We prove, as a lemma, that $(\ast)$ never happens.

Secondly, for correctness (\cref{thm:at-cor}), our conversion of runs should replace an invocation of the memoization table---if it returns a failure with a shallower depth---with not only the corresponding run (as before) but also the run of the second recursive call (Line~19). \FinalOrArxiv{See~\cite[Appendix~B.3]{FujinamiH24ESOP_arxiv_extended_ver}}{See \cref{appendix:memo-at-proof}} for details.

\myparagraph{Combination with Look-around}\label{subsec:combWithLA}
It is also possible to combine with \cref{alg:la-memoization} (for look-around) and \cref{alg:at-memoization} (for atomic grouping).
In this case, the type of memoization tables becomes $M\colon P \times \mathbb{N} \rightharpoonup \{ \mathsf{Failure}(j)\mid j \in \{ 0, \dots, \nu(\mathcal{A}) \} \} \cup \{ \mathsf{Success} \}$ and nesting depths of the atomic group are reset by look-around operators.
A complete algorithm can be found \FinalOrArxiv{in~\cite[Appendix~C]{FujinamiH24ESOP_arxiv_extended_ver}}{in \cref{appendix:la-at-memoization}}; it also exhibits the desired properties.

\section{Experiments and Evaluation}\label{sec:exp}

\myparagraph{Implementation}
We implemented the algorithm proposed in this paper for evaluation.
We call our implementation \texttt{memo-regex}.
It is written in 1368 lines of Scala.

\texttt{memo-regex} supports both look-around (i.e., look-ahead and look-behind) and atomic grouping.
We implemented a regex parser ourselves.
Backtracking is implemented by managing a stack manually rather than using a recursive function to prevent stack overflow.
In this case, the memoization keys are pushed onto the stack.
Recoding these keys in a memoization table is done during backtracking.
We used the mutable \texttt{HashMap} from the Scala standard library as a data structure for memoization tables.

\texttt{memo-regex} also supports capturing sub-matchings.
However, this feature cannot be used within atomic grouping and positive look-around because sub-matching information is lost for memoization.

The code of \texttt{memo-regex}, as well as all experiment scripts, is available~\cite{fujinami_2024_10458317}.

\myparagraph{Efficiency of Our Algorithm}
We conducted experiments to assess the performance of our \texttt{memo-regex}, in particular in comparison with other existing implementations. 

\begin{table}[tb]
  \caption{our benchmark regexes and input strings}\label{tab:redos-vuln-regexes}
  \begin{tabular}{rl}\toprule
    $r_1$ &
    \begin{tabular}{l}
      \EscVerb{/^(?=^.{1,254}\$)(^(?:(?!\\.|-)([a-z0-9\\-\\*]{1,63}|([a-z0-9\\-]{1,62}[a-} \\
      \EscVerb{z0-9]))\\.)+(?:[a-z]{2,})$)$/s} \\
      \scriptsize{\qquad}input: $w_1 = \texttt{"$0.$" "$0.0a.$"${}^n$ "${\backslash}u0000$"}$, complexity: $O(2^n)$ \\
      \scriptsize{\qquad}\url{https://regexlib.com/REDetails.aspx?regexp_id=3494}
    \end{tabular} \\
    \midrule
    $r_2$ &
    \begin{tabular}{l}
      \EscVerb{/(?=(?:[^\\']*\\'[^\\']*\\')*(?![^\\']*\\'))/} \\
      \scriptsize{\qquad}input: $w_2 = \texttt{"$x$"${}^n$ "$'$"}$, complexity: $O(n^2)$ \\
      \scriptsize{\qquad}\url{https://regexlib.com/REDetails.aspx?regexp_id=938}
    \end{tabular} \\
    \midrule
    $r_3$ &
    \begin{tabular}{l}
      \EscVerb{/(?<=[\\w\\s](?:[\\.\\!\\? ]+[\\x20]*[\\x22\\xBB]*))(?:\\s+(?![\\x22\\xBB](?!\\w)} \\
      \EscVerb{))/} \\
      \scriptsize{\qquad}input: $w_3 = \texttt{"${\backslash}"$"${}^n$ "$\hspace{1em}$"}$, complexity: $O(n^2)$ \\
      \scriptsize{\qquad}\url{https://regexlib.com/REDetails.aspx?regexp_id=2355}
    \end{tabular} \\
    \midrule
    $r_4$ &
    \begin{tabular}{l}
      \EscVerb{/(?:(<)\\s*?(\\w+)(\\s*?(?>(?!=[\\/\\?]?>)(\\w+)(?:\\s*(=)\\s*)((?:\\'[^\\']*\\'} \\
      \EscVerb{|\\"[^\\"]*\\"|[^ >]+))))\\s*?([\\/\\?]?>))/} \\
      \scriptsize{\qquad}input: $w_4 = \texttt{"$<$" "$aaa$"${}^n$ "$>$"}$, complexity: $O(n^2)$ \\
      \scriptsize{\qquad}\url{https://regexlib.com/REDetails.aspx?regexp_id=373}
    \end{tabular} \\
    \bottomrule
  \end{tabular}
\end{table}

As target regexes, we looked for those with look-around and/or atomic grouping in the real-world regexes posted on \url{regexlib.com}. We then identified---by manual inspection---four regexes $r_1,\dotsc, r_4$ that are subject to potential catastrophic backtracking. These regexes are shown in \cref{tab:redos-vuln-regexes}. We then crafted input strings $w_1, \dotsc, w_4$, respectively, so that they cause catastrophic backtracking.
Specifically, $r_1$ contains positive look-ahead and negative look-ahead;
this positive look-ahead is used for restricting the length of input strings.
The regexes $r_2$ and $r_3$ are themselves positive look-ahead and look-behind, respectively; both include negative look-ahead, too.
The regex $r_4$ includes atomic grouping and negative look-ahead.

For these regexes, we measured matching time using \texttt{memo-regex} on OpenJDK 21.0.1. We compared it with the following implementations: Node.js 20.5.0, Ruby 3.1.4, and  PCRE2 10.42 (used by PHP 8.3.1, w/ or w/o JIT).
All of these implementations use backtracking; Ruby and PCRE2 have restrictions on regexes inside look-behind and Node.js does not support atomic grouping.
The experiments were performed 10 times and the average was adopted.
Furthermore, for \texttt{memo-regex}, we measured the size of its memoization table by the memory usage, using jamm.\footnote{\url{https://github.com/jbellis/jamm}}
The experiments were conducted on MacBook Pro 2021 (Apple M1 Pro, RAM: 32 GB).


\begin{figure}[tb]
  \centering
  \begin{minipage}[t]{0.3\textwidth}
    \includegraphics[width={\hsize}]{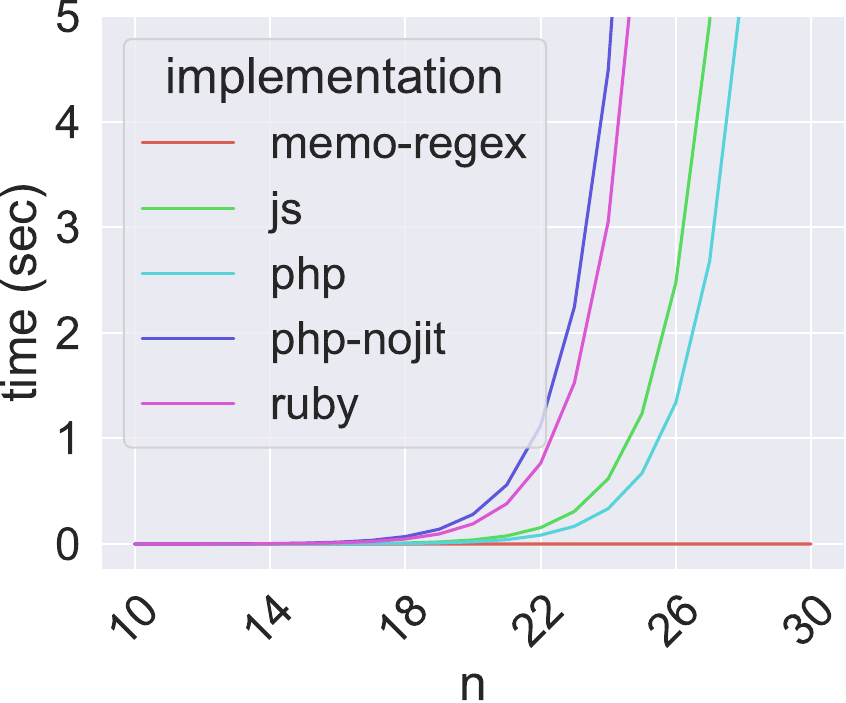}
    \subcaption{comparison to others}
  \end{minipage}
  \begin{minipage}[t]{0.3\textwidth}
    \includegraphics[width={\hsize}]{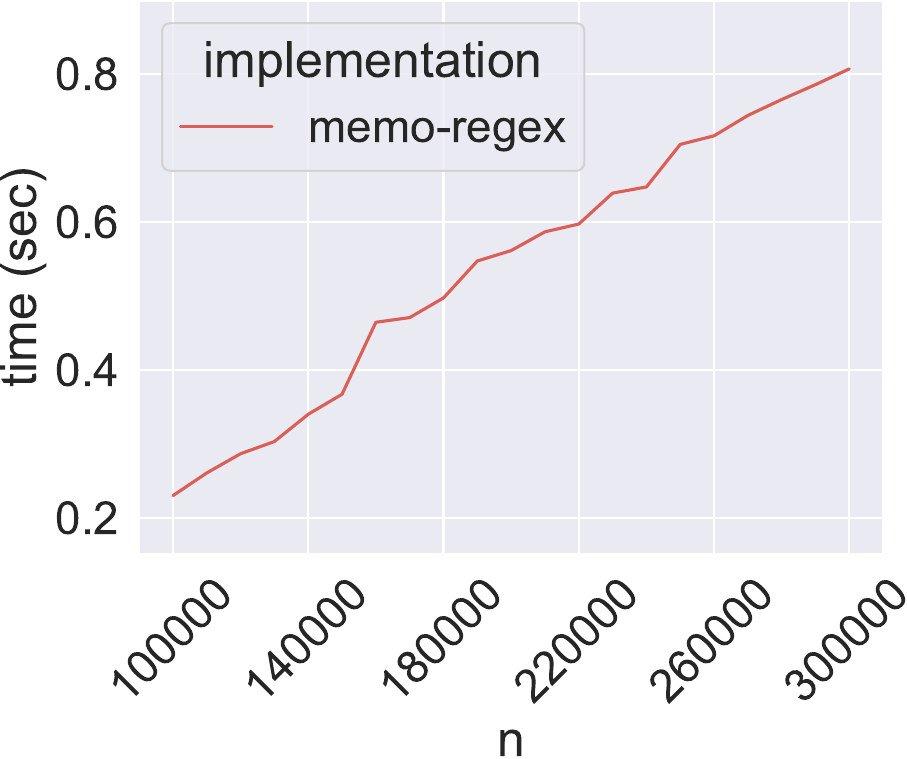}
    \subcaption{\texttt{memo-regex} only}
  \end{minipage}
  \begin{minipage}[t]{0.3\textwidth}
    \includegraphics[width={\hsize}]{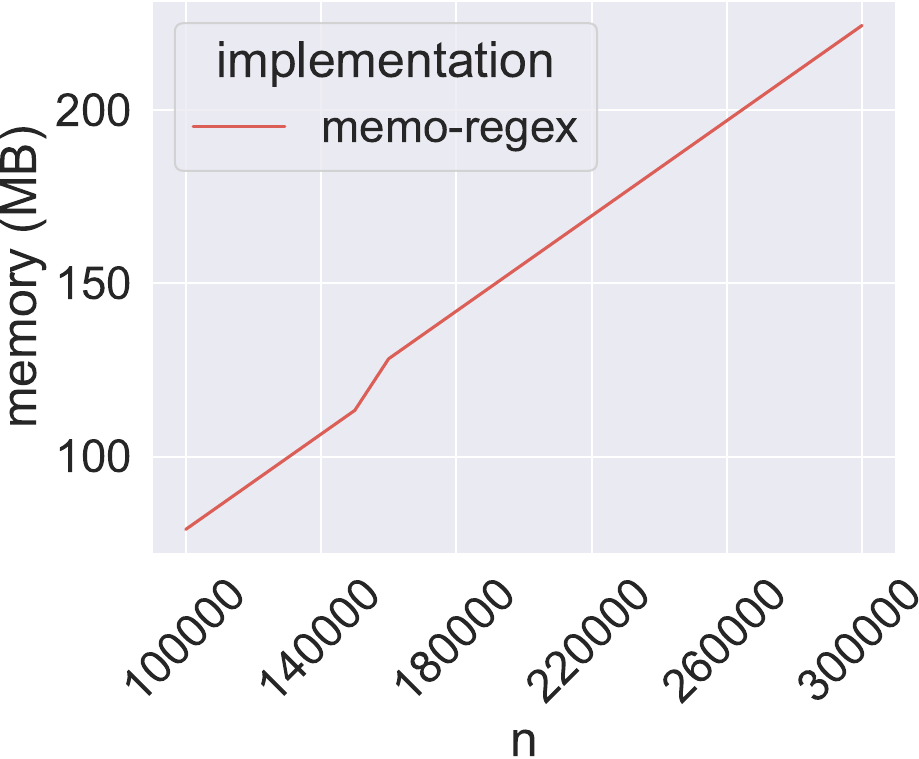}
    \subcaption{memoization table size}
  \end{minipage}
  \caption{result for $r_1$}\label{fig:matching-time-r1}
\end{figure}

\begin{figure}[tb]
  \centering
  \begin{minipage}[b]{0.3\textwidth}
    \includegraphics[width={\hsize}]{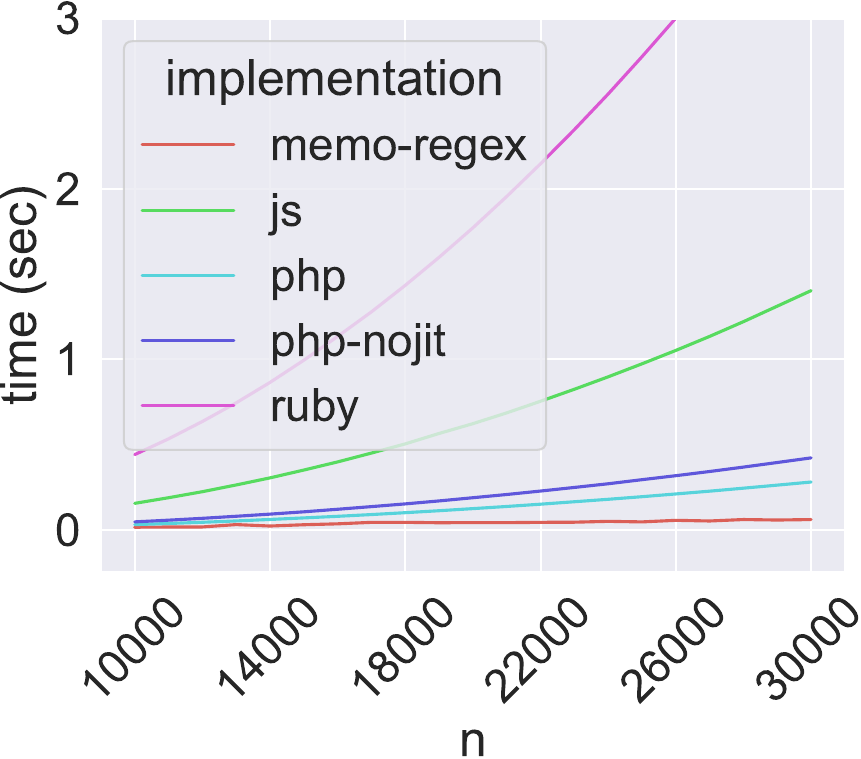}
    \subcaption{$r_2$}
  \end{minipage}
  \begin{minipage}[b]{0.3\textwidth}
    \includegraphics[width={\hsize}]{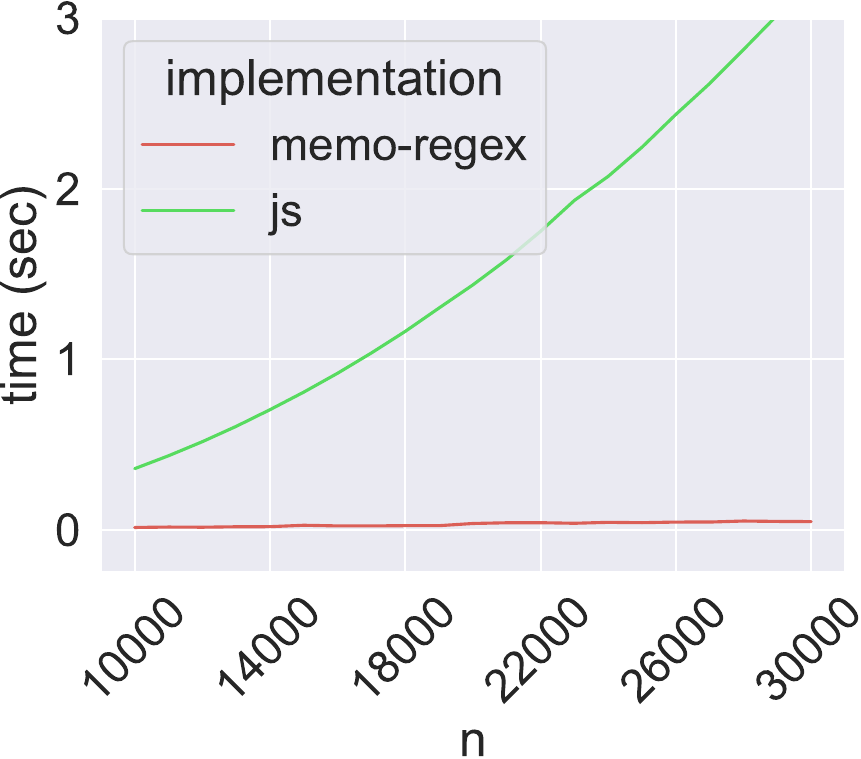}
    \subcaption{$r_3$}
  \end{minipage}
  \begin{minipage}[b]{0.3\textwidth}
    \includegraphics[width={\hsize}]{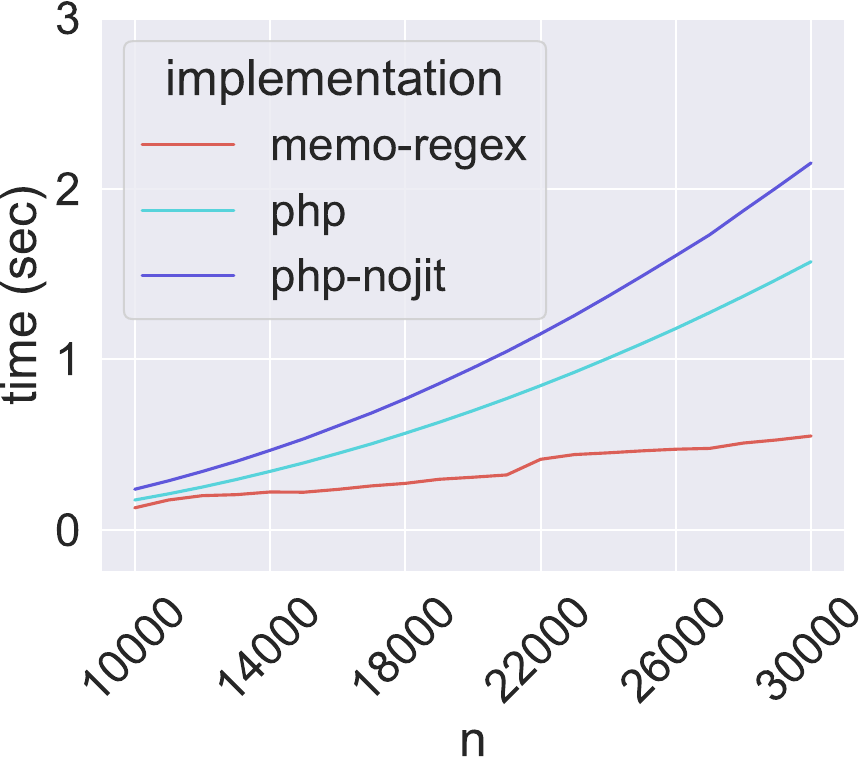}
    \subcaption{$r_4$}
  \end{minipage}
  \caption{matching time for $r_2, r_3$ and $r_4$}\label{fig:matching-time-r2-4}
\end{figure}


We show the results in \cref{fig:matching-time-r1,fig:matching-time-r2-4}.
Note that the values of $n$ are different depending on whether the matching time complexity is $O(n^2)$ or $O(2^n)$. Results for some implementations are absent for $r_{3}$ and $r_4$ because of the syntactic restrictions discussed above. 

In \cref{fig:matching-time-r1,fig:matching-time-r2-4},
we observe clear performance advantages of $\texttt{memo-regex}$. In particular, its linear-time complexity and linear memory consumption (memoization table size) are experimentally confirmed.

\myparagraph{Real-world Usage of Look-around and Atomic Grouping
}
We additionally surveyed the use of the regex extensions of our interest, in order to confirm their practical relevance.

We used a  regex dataset collected by a 2019 survey~\cite{DBLP:conf/sigsoft/DavisMCSL19}.
This dataset contains 537,806 regexes collected from the source code of real-world products.

We tallied the usage of each regex extension by parsing these regexes in the dataset with our parser in \texttt{memo-regex}.
8,679 regexes could not be parsed or compiled; this is due to back-reference for 4,360 regexes, unsupported syntax (Unicode character class, conditional branching, etc.) for 4,134 regexes, and too large or semantically invalid regexes for the other 184 regexes.
We adopted the remaining 529,127 regexes for tallying.

\begin{wraptable}[10]{r}[0pt]{5cm}
  \vspace{-2.2em}

  \caption{regex ext.\ usage}\label{tab:usage-regex-extension}
  \centering
  \scalebox{.8}{  \begin{tabular}{lr}\toprule
      feature              & \# of regexes \\
      \midrule
      (total)              & 529,127 \phantom{(0.0\%)} \\
      positive look-ahead  & 7,476 (1.4\%) \\
      negative look-ahead  & 6,917 (1.3\%) \\
      positive look-behind & 1,746 (0.3\%) \\
      negative look-behind & 3,750 (0.7\%) \\
      atomic grouping      & 1,113 (0.2\%) \\
      at least one of the above & 17,167 (3.2\%)\\\bottomrule
    \end{tabular}}
\end{wraptable}
The result is shown in \cref{tab:usage-regex-extension}.
Note that \begin{enumerate*}[label=\arabic*)]
  \item the numbers for look-ahead and look-behind do not include simple zero-width assertions such as \texttt{\^} (line-begin) or \texttt{\$} (line-end), and
  \item that of atomic grouping includes possessive quantifiers such as \texttt{*+} and \texttt{++}.
\end{enumerate*}

In \cref{tab:usage-regex-extension}, we observe that 17,167 regexes (3.2\%) in the dataset use at least one of the extensions we studied in this paper. While the ratio is not very large, the absolute number (17,167 regexes) is significant; this implies that there are a number of applications (such as web services) that rely on the regex extensions. Thereby we confirm the practical relevance of these regex extensions.

\section{Conclusions and Future Work}\label{sec:conclusion}

In this paper, we proposed a backtracking algorithm with memoization for regexes with look-around and atomic grouping.
It is the first linear-time backtracking matching algorithm for such regexes.
It also fixs problems of the memoization matching algorithm in \cite{DBLP:conf/sp/DavisSL21} for look-ahead.
We implemented the algorithm; our experimental evaluation confirms its performance advantage.

One direction of future work is to support more extensions.
Our implementation 
does not support a widely used regex extension, namely back-references.
In the recent work~\cite{DBLP:conf/sp/DavisSL21}, back-reference was supported by additionally recording captured positions in memoization tables. We expect that a similar idea is applicable to our algorithm. 

Combination with selective memoization (used in~\cite{DBLP:conf/sp/DavisSL21}; \FinalOrArxiv{see~\cite[Appendix~A]{FujinamiH24ESOP_arxiv_extended_ver}}{see \cref{appendix:selectiveMemo}}) is another direction.
We believe it is possible,
but it will require a more detailed discussion on how to handle sub-automata in the selective memoization schema.


\section*{Acknowledgments}
Thanks are due to Konstantinos Mamouras for pointing to~\cite{MamourasC24} after the dissemination of the preprint version of this paper.

\section*{Data-Availability Statement}
The data that support the findings of this study are openly available in Zenodo at \url{10.5281/zenodo.10458317}, reference number~\cite{fujinami_2024_10458317}.

\bibliographystyle{splncs04}
\bibliography{main}

\begin{thebibliography}{10}
\providecommand{\url}[1]{\texttt{#1}}
\providecommand{\urlprefix}{URL }
\providecommand{\doi}[1]{https://doi.org/#1}

\bibitem{DBLP:books/el/leeuwen90/Aho90}
Aho, A.V.: Algorithms for finding patterns in strings. In: van Leeuwen, J.
  (ed.) Handbook of Theoretical Computer Science, Volume {A:} Algorithms and
  Complexity, pp. 255--300. Elsevier and {MIT} Press (1990)

\bibitem{DBLP:journals/corr/BerglundDM14}
Berglund, M., Drewes, F., van~der Merwe, B.: Analyzing catastrophic
  backtracking behavior in practical regular expression matching. In:
  {\'{E}}sik, Z., F{\"{u}}l{\"{o}}p, Z. (eds.) Proceedings 14th International
  Conference on Automata and Formal Languages, {AFL} 2014, Szeged, Hungary, May
  27-29, 2014. {EPTCS}, vol.~151, pp. 109--123 (2014).
  \doi{10.4204/EPTCS.151.7}, \url{https://doi.org/10.4204/EPTCS.151.7}

\bibitem{DBLP:journals/jucs/BerglundML21}
Berglund, M., van~der Merwe, B., van Litsenborgh, S.: Regular expressions with
  lookahead. J. Univers. Comput. Sci.  \textbf{27}(4),  324--340 (2021).
  \doi{10.3897/jucs.66330}, \url{https://doi.org/10.3897/jucs.66330}

\bibitem{DBLP:conf/wia/BerglundMWW17}
Berglund, M., van~der Merwe, B., Watson, B.W., Weideman, N.: On the semantics
  of atomic subgroups in practical regular expressions. In: Carayol, A.,
  Nicaud, C. (eds.) Implementation and Application of Automata - 22nd
  International Conference, {CIAA} 2017, Marne-la-Vall{\'{e}}e, France, June
  27-30, 2017, Proceedings. Lecture Notes in Computer Science, vol. 10329, pp.
  14--26. Springer (2017). \doi{10.1007/978-3-319-60134-2\_2},
  \url{https://doi.org/10.1007/978-3-319-60134-2\_2}

\bibitem{DBLP:conf/fscd/ChidaT22}
Chida, N., Terauchi, T.: On lookaheads in regular expressions with
  backreferences. In: Felty, A.P. (ed.) 7th International Conference on Formal
  Structures for Computation and Deduction, {FSCD} 2022, August 2-5, 2022,
  Haifa, Israel. LIPIcs, vol.~228, pp. 15:1--15:18. Schloss Dagstuhl -
  Leibniz-Zentrum f{\"{u}}r Informatik (2022).
  \doi{10.4230/LIPIcs.FSCD.2022.15},
  \url{https://doi.org/10.4230/LIPIcs.FSCD.2022.15}

\bibitem{DBLP:conf/sp/ChidaT22}
Chida, N., Terauchi, T.: Repairing dos vulnerability of real-world regexes. In:
  43rd {IEEE} Symposium on Security and Privacy, {SP} 2022, San Francisco, CA,
  USA, May 22-26, 2022. pp. 2060--2077. {IEEE} (2022).
  \doi{10.1109/SP46214.2022.9833597},
  \url{https://doi.org/10.1109/SP46214.2022.9833597}

\bibitem{cox2007}
Cox, R.: Regular expression matching can be simple and fast (but is slow in
  java, perl, php, python, ruby,...) (2007), {URL}:
  \url{http://swtch.com/rsc/regexp/regexp1.html}

\bibitem{cox2009}
Cox, R.: Regular expression matching: the virtual machine approach (2009),
  {URL}: \url{http://swtch.com/~rsc/regexp/regexp2.html}

\bibitem{DBLP:conf/sigsoft/DavisMCSL19}
Davis, J.C., IV, L.G.M., Coghlan, C.A., Servant, F., Lee, D.: Why aren't
  regular expressions a lingua franca? an empirical study on the re-use and
  portability of regular expressions. In: Dumas, M., Pfahl, D., Apel, S.,
  Russo, A. (eds.) Proceedings of the {ACM} Joint Meeting on European Software
  Engineering Conference and Symposium on the Foundations of Software
  Engineering, {ESEC/SIGSOFT} {FSE} 2019, Tallinn, Estonia, August 26-30, 2019.
  pp. 443--454. {ACM} (2019). \doi{10.1145/3338906.3338909},
  \url{https://doi.org/10.1145/3338906.3338909}

\bibitem{DBLP:conf/sp/DavisSL21}
Davis, J.C., Servant, F., Lee, D.: Using selective memoization to defeat
  regular expression denial of service (redos). In: 42nd {IEEE} Symposium on
  Security and Privacy, {SP} 2021, San Francisco, CA, USA, 24-27 May 2021. pp.
  1--17. {IEEE} (2021). \doi{10.1109/SP40001.2021.00032},
  \url{https://doi.org/10.1109/SP40001.2021.00032}

\bibitem{davis2021impl}
Davis, J.: {PurdueDualityLab}/memoized-regex-engine: {IEEE S\&P\'21} artifact
  (Apr 2021). \doi{10.5281/zenodo.4718966}, {URL}:
  \url{https://zenodo.org/record/4718966}

\bibitem{stackexchange2016}
Exchange, S.: Stack exchange network status — outage postmortem - july 20,
  2016. URL:
  \url{http://web.archive.org/web/20210308040819/https://stackstatus.net/post/147710624694/outage-postmortem-july-20-2016}
  (2016)

\bibitem{DBLP:books/daglib/0016809}
Friedl, J.E.F.: Mastering regular expressions - understand your data and be
  more productive: for Perl, PHP, Java, .NET, Ruby, and more {(3.} ed.).
  O'Reilly (2006), \url{http://www.oreilly.de/catalog/regex3/index.html}

\bibitem{fujinami_2024_10458317}
Fujinami, H., Hasuo, I.: {Artifact Archive for memo-regex: linear-time regex
  matching implementation with memoization} (Jan 2024).
  \doi{10.5281/zenodo.10458317}

\bibitem{cloudflare2019}
Graham-Cumming, J.: Details of the cloudflare outage on july 2, 2019 (2019),
  {URL}:
  \url{https://web.archive.org/web/20190712160002/https://blog.cloudflare.com/details-of-the-cloudflare-outage-on-july-2-2019/}

\bibitem{DBLP:conf/cgo/Herczeg14}
Herczeg, Z.: Extending the {PCRE} library with static backtracking based
  just-in-time compilation support. In: Kaeli, D.R., Moseley, T. (eds.) 12th
  Annual {IEEE/ACM} International Symposium on Code Generation and
  Optimization, {CGO} 2014, Orlando, FL, USA, February 15-19, 2014. p.~306.
  {ACM} (2014), \url{https://dl.acm.org/citation.cfm?id=2544146}

\bibitem{DBLP:conf/nss/KirrageRT13}
Kirrage, J., Rathnayake, A., Thielecke, H.: Static analysis for regular
  expression denial-of-service attacks. In: L{\'{o}}pez, J., Huang, X., Sandhu,
  R.S. (eds.) Network and System Security - 7th International Conference, {NSS}
  2013, Madrid, Spain, June 3-4, 2013. Proceedings. Lecture Notes in Computer
  Science, vol.~7873, pp. 135--148. Springer (2013).
  \doi{10.1007/978-3-642-38631-2\_11},
  \url{https://doi.org/10.1007/978-3-642-38631-2\_11}

\bibitem{DBLP:conf/uss/LiCC0PCCC21}
Li, Y., Chen, Z., Cao, J., Xu, Z., Peng, Q., Chen, H., Chen, L., Cheung, S.:
  {ReDoSHunter}: {A} combined static and dynamic approach for regular
  expression dos detection. In: Bailey, M., Greenstadt, R. (eds.) 30th {USENIX}
  Security Symposium, {USENIX} Security 2021, August 11-13, 2021. pp.
  3847--3864. {USENIX} Association (2021),
  \url{https://www.usenix.org/conference/usenixsecurity21/presentation/li-yeting}

\bibitem{DBLP:conf/uss/LiS0CLLCC0X22}
Li, Y., Sun, Y., Xu, Z., Cao, J., Li, Y., Li, R., Chen, H., Cheung, S., Liu,
  Y., Xiao, Y.: {RegexScalpel}: Regular expression denial of service (redos)
  defense by localize-and-fix. In: Butler, K.R.B., Thomas, K. (eds.) 31st
  {USENIX} Security Symposium, {USENIX} Security 2022, Boston, MA, USA, August
  10-12, 2022. pp. 4183--4200. {USENIX} Association (2022),
  \url{https://www.usenix.org/conference/usenixsecurity22/presentation/li-yeting}

\bibitem{DBLP:conf/kbse/LiXCCGCZ20}
Li, Y., Xu, Z., Cao, J., Chen, H., Ge, T., Cheung, S., Zhao, H.: {FlashRegex}:
  Deducing anti-redos regexes from examples. In: 35th {IEEE/ACM} International
  Conference on Automated Software Engineering, {ASE} 2020, Melbourne,
  Australia, September 21-25, 2020. pp. 659--671. {IEEE} (2020).
  \doi{10.1145/3324884.3416556}, \url{https://doi.org/10.1145/3324884.3416556}

\bibitem{MamourasC24}
Mamouras, K., Chattopadhyay, A.: Efficient matching of regular expressions with
  lookaround assertions. Proc. {ACM} Program. Lang.  \textbf{8}({POPL}),
  2761--2791 (2024). \doi{10.1145/3632934},
  \url{https://doi.org/10.1145/3632934}

\bibitem{DBLP:journals/tc/McNaughtonY60}
McNaughton, R., Yamada, H.: Regular expressions and state graphs for automata.
  {IRE} Trans. Electron. Comput.  \textbf{9}(1),  39--47 (1960).
  \doi{10.1109/TEC.1960.5221603},
  \url{https://doi.org/10.1109/TEC.1960.5221603}

\bibitem{DBLP:conf/wia/MerweMLB21}
van~der Merwe, B., Mouton, J., van Litsenborgh, S., Berglund, M.: Memoized
  regular expressions. In: Maneth, S. (ed.) Implementation and Application of
  Automata - 25th International Conference, {CIAA} 2021, Virtual Event, July
  19-22, 2021, Proceedings. Lecture Notes in Computer Science, vol. 12803, pp.
  39--52. Springer (2021). \doi{10.1007/978-3-030-79121-6\_4},
  \url{https://doi.org/10.1007/978-3-030-79121-6\_4}

\bibitem{DBLP:conf/kbse/MichaelDDLS19}
Michael, IV, L.G., Donohue, J., Davis, J.C., Lee, D., Servant, F.: Regexes are
  hard: Decision-making, difficulties, and risks in programming regular
  expressions. In: 34th {IEEE/ACM} International Conference on Automated
  Software Engineering, {ASE} 2019, San Diego, CA, USA, November 11-15, 2019.
  pp. 415--426. {IEEE} (2019). \doi{10.1109/ASE.2019.00047},
  \url{https://doi.org/10.1109/ASE.2019.00047}

\bibitem{DBLP:journals/jip/MiyazakiM19}
Miyazaki, T., Minamide, Y.: Derivatives of regular expressions with lookahead.
  J. Inf. Process.  \textbf{27},  422--430 (2019).
  \doi{10.2197/ipsjjip.27.422}, \url{https://doi.org/10.2197/ipsjjip.27.422}

\bibitem{morihata2012}
Morihata, A.: Translation of regular expression with lookahead into finite
  state automaton. Computer Software  \textbf{29}(1),  1\_147--1\_158 (2012).
  \doi{10.11309/jssst.29.1_147}

\bibitem{DBLP:journals/pacmpl/MoseleyNRSTVWX23}
Moseley, D., Nishio, M., Rodriguez, J.P., Saarikivi, O., Toub, S., Veanes, M.,
  Wan, T., Xu, E.: Derivative based nonbacktracking real-world regex matching
  with backtracking semantics. Proc. {ACM} Program. Lang.  \textbf{7}({PLDI}),
  1026--1049 (2023). \doi{10.1145/3591262},
  \url{https://doi.org/10.1145/3591262}

\bibitem{DBLP:journals/corr/RathnayakeT14}
Rathnayake, A., Thielecke, H.: Static analysis for regular expression
  exponential runtime via substructural logics. CoRR  \textbf{abs/1405.7058}
  (2014), \url{http://arxiv.org/abs/1405.7058}

\bibitem{DBLP:journals/japll/SakumaMV12}
Sakuma, Y., Minamide, Y., Voronkov, A.: Translating regular expression matching
  into transducers. J. Appl. Log.  \textbf{10}(1),  32--51 (2012).
  \doi{10.1016/j.jal.2011.11.003},
  \url{https://doi.org/10.1016/j.jal.2011.11.003}

\bibitem{DBLP:conf/kbse/Shen000ML18}
Shen, Y., Jiang, Y., Xu, C., Yu, P., Ma, X., Lu, J.: {ReScue}: crafting regular
  expression dos attacks. In: Huchard, M., K{\"{a}}stner, C., Fraser, G. (eds.)
  Proceedings of the 33rd {ACM/IEEE} International Conference on Automated
  Software Engineering, {ASE} 2018, Montpellier, France, September 3-7, 2018.
  pp. 225--235. {ACM} (2018). \doi{10.1145/3238147.3238159},
  \url{https://doi.org/10.1145/3238147.3238159}

\bibitem{DBLP:conf/fccm/SidhuP01}
Sidhu, R.P.S., Prasanna, V.K.: Fast regular expression matching using fpgas.
  In: The 9th Annual {IEEE} Symposium on Field-Programmable Custom Computing
  Machines, {FCCM} 2001, Rohnert Park, California, USA, April 29 - May 2, 2001.
  pp. 227--238. {IEEE} Computer Society (2001). \doi{10.1109/FCCM.2001.22},
  \url{https://doi.ieeecomputersociety.org/10.1109/FCCM.2001.22}

\bibitem{DBLP:conf/uss/StaicuP18}
Staicu, C., Pradel, M.: Freezing the web: {A} study of redos vulnerabilities in
  javascript-based web servers. In: Enck, W., Felt, A.P. (eds.) 27th {USENIX}
  Security Symposium, {USENIX} Security 2018, Baltimore, MD, USA, August 15-17,
  2018. pp. 361--376. {USENIX} Association (2018),
  \url{https://www.usenix.org/conference/usenixsecurity18/presentation/staicu}

\bibitem{sugiyama2014}
Sugiyama, S., Minamide, Y.: Checking time linearity of regular expression
  matching based on backtracking. Information and Media Technologies
  \textbf{9}(3),  222--232 (2014). \doi{10.11185/imt.9.222}

\bibitem{DBLP:journals/cacm/Thompson68}
Thompson, K.: Regular expression search algorithm. Commun. {ACM}
  \textbf{11}(6),  419--422 (1968). \doi{10.1145/363347.363387},
  \url{https://doi.org/10.1145/363347.363387}

\bibitem{DBLP:conf/wia/WeidemanMBW16}
Weideman, N., van~der Merwe, B., Berglund, M., Watson, B.W.: Analyzing matching
  time behavior of backtracking regular expression matchers by using ambiguity
  of {NFA}. In: Han, Y., Salomaa, K. (eds.) Implementation and Application of
  Automata - 21st International Conference, {CIAA} 2016, Seoul, South Korea,
  July 19-22, 2016, Proceedings. Lecture Notes in Computer Science, vol.~9705,
  pp. 322--334. Springer (2016). \doi{10.1007/978-3-319-40946-7\_27},
  \url{https://doi.org/10.1007/978-3-319-40946-7\_27}

\bibitem{DBLP:conf/tacas/WustholzOHD17}
W{\"{u}}stholz, V., Olivo, O., Heule, M.J.H., Dillig, I.: Static detection of
  dos vulnerabilities in programs that use regular expressions. In: Legay, A.,
  Margaria, T. (eds.) Tools and Algorithms for the Construction and Analysis of
  Systems - 23rd International Conference, {TACAS} 2017, Held as Part of the
  European Joint Conferences on Theory and Practice of Software, {ETAPS} 2017,
  Uppsala, Sweden, April 22-29, 2017, Proceedings, Part {II}. Lecture Notes in
  Computer Science, vol. 10206, pp. 3--20 (2017).
  \doi{10.1007/978-3-662-54580-5\_1},
  \url{https://doi.org/10.1007/978-3-662-54580-5\_1}

\end{thebibliography}

\appendix
\begin{ArxivVersion}

\section{Selective Memoization}\label{appendix:selectiveMemo}
We describe the second main contribution of the work~\cite{DBLP:conf/sp/DavisSL21},
called \emph{selective memoization}. It shrinks memoization tables and thus reduces memory usage. 

\emph{Selective memoization} is a general technique in which one chooses not to record certain entries, especially in such a case that those entries are easily recoverable from other entries. A famous example is for the Fibonacci function $\mathsf{Fib}(n)$: even if one chooses not to record the values of $\mathsf{Fib}(3n+2)$ for each $n$, it is easily computed by the recorded values $\mathsf{Fib}(3n), \mathsf{Fib}(3n+1)$. The work~\cite{DBLP:conf/sp/DavisSL21} pursued the same in the context of the regex matching algorithm in~\cref{ssec:davisAlgo}. 

Specifically, in~\cite{DBLP:conf/sp/DavisSL21}, the set $Q$ in the original memoization table $Q\times \mathbb{N}\rightharpoonup \{\mathbf{false}\}$ (cf.\ \cref{ssec:davisAlgo}) is restricted to a set $Q_{\textsf{sel}}$. For the last set $Q_{\textsf{sel}}$, the following two options are proposed.
\begin{itemize}
 \item $Q_{\textsf{sel}} = Q_{\textsf{in-deg}>1}$ is the set of states where the number of incoming transitions is at least 2.
 \item $Q_{\textsf{sel}} = Q_\textsf{anc}$ is the set of states that are ancestors of loops when states are topologically sorted by transitions from the initial state.
\end{itemize}
The choices of these restrictions are a matter of heuristics; nevertheless, 
it is established in~\cite[Section~VII]{DBLP:conf/sp/DavisSL21} that these restrictions (and the resulting selective memoization) still suffice for linear-time backtracking matching. 

The work~\cite{DBLP:conf/wia/MerweMLB21} further develops the theory by 1) showing that finding a minimal restriction of $Q$ that achieves linear-time backtracking matching in \cref{alg:davis-matching} is NP-complete, and 2) proposing more sophisticated restriction heuristics (i.e., choices of $Q_{\textsf{sel}}$) based on the domain knowledge of ReDoS.

\section{Proofs for \crefrange{thm:algo-lin}{thm:at-cor}}\label{appendix:proof}

First, recall that we are dealing with five algorithms.
\begin{itemize}
  \item $\textproc{Match}$ from \cref{alg:matching} is a matching algorithm for (non-extended) NFAs and $\textproc{\MatchLaAt}$ from \cref{alg:la-at-matching} is a matching algorithm for extended NFAs.
  \item $\textproc{Memo}$ from \cref{alg:memoization} is a matching algorithm with memoization for (non-extended) NFAs.
  \item $\textproc{\MemoLa}$ from \cref{alg:la-memoization} is a matching algorithm with memoization for la-NFAs and $\textproc{\MemoAt}$ from \cref{alg:at-memoization} is a matching algorithm with memoization for at-NFAs.
\end{itemize}

In \cref{appendix:memo-proof}, we give proofs for the linear-time complexity (\cref{thm:algo-lin}) and the correctness (\cref{thm:algo-cor}) of $\textproc{Memo}$.
Next, in \cref{appendix:memo-la-proof}, we also give proofs for the linear-time complexity (\cref{thm:la-lin}) and the correctness (\cref{thm:la-cor}) of $\textproc{\MemoLa}$, and in \cref{appendix:memo-at-proof}, give proofs for the linear-time complexity (\cref{thm:at-lin}) and the correctness (\cref{thm:at-cor}) of $\textproc{\MemoAt}$.

For any algorithm, our proof strategies are as follows.
\begin{itemize}
  \item We first define the notions of \emph{run} for the $\textproc{Match}$ and $\textproc{Memo}$ algorithms. A run is an execution trace that contains information about recursive calls and their result values.
  \item To show the linear-time complexity of the $\textproc{Memo}$ algorithm, we show the length of a run is bounded by $O(|w|)$, where $|w|$ is the length of the string.
  \item To show the correctness of the $\textproc{Memo}$ algorithm, we introduce a conversion method from a run of $\textproc{Memo}$ to a run of $\textproc{Match}$ while preserving the matching result.
\end{itemize}

\subsection{Proofs for the Linear-time Complexity and Correctness of $\textproc{Memo}$ from \cref{alg:memoization} (\cref{thm:algo-lin} and \cref{thm:algo-cor})}\label{appendix:memo-proof}

First, we define runs of both $\textproc{Match}$ and $\textproc{Memo}$.

\begin{definition}[a run of $\textproc{Match}$]\label{def:match-run}
  For an NFA $\mathcal{A} = (Q, q_0, F, T)$, an input string $w$, we fix a set $C = \{ \mathsf{C}(q, i) \mid q \in Q, i \in \{ 0, \dots, |w| \} \} \cup \{ \mathsf{S}(i') \mid i' \in \{ 0, \dots, |w| \} \} \cup \{ \mathsf{F} \}$.
  For a position $i_0 \in \{ 0, \dots, |w| \}$, we define a run $s$ of $\Call{\MatchAw}{q_0, i_0}$ as a sequence of $C$, which is crafted in the following manner:
  \begin{itemize}
    \item The initial value of $s$ is the empty sequence, and it grows on the execution of $\Call{\MatchAw}{q_0, i_0}$.
    \item When $\textproc{Match}$ is called with some $q \in Q, i \in \{ 0, \dots, |w| \}$, $\mathsf{C}(q, i)$ is appended to $s$.
    \item When $\textproc{Match}$ returns $\mathsf{SuccessAt}(i')$ for some $i' \in \{ 0, \dots, |w| \}$, $\mathsf{S}(i')$ is appended to $s$.
    \item When $\textproc{Match}$ returns $\mathsf{Failure}$, $\mathsf{F}$ is appended to $s$.
  \end{itemize}
\end{definition}

\begin{definition}[a run of $\textproc{Memo}$]\label{def:memo-run}
  For an NFA $\mathcal{A} = (Q, q_0, F, T)$, an input string $w$, we fix a set $C' = C \cup \{ \mathsf{M}(q, i) \mid q \in Q, i \in \{ 0, \dots, |w| \} \}$.
  For a position $i_0 \in \{ 0, \dots, |w| \}$, we define a run $s$ of $\Call{\MemoMzAw}{q_0, i_0}$ as a sequence of $C'$, which is crafted in the following manner:
  \begin{itemize}
    \item The initial value of $s$ is the empty sequence, and it grows on the execution of $\Call{\MemoMzAw}{q_0, i_0}$.
    \item When $\textproc{Memo}$ is called with some $q \in Q, i \in \{ 0, \dots, |w| \}$, but $\mathsf{M}(q, i)$ is not defined, $\mathsf{C}(q, i)$ is appended to $s$.
    \item When $\textproc{Memo}$ is called with some $q \in Q, i \in \{ 0, \dots, |w| \}$ and $\mathsf{M}(q, i)$ is defined, $\mathsf{M}(q, i)$ is appended to $s$.
    \item When $\textproc{Memo}$ returns $\mathsf{SuccessAt}(i')$ for some $i' \in \{ 0, \dots, |w| \}$, $\mathsf{S}(i')$ is appended to $s$.
    \item When $\textproc{Memo}$ returns $\mathsf{Failure}$, $\mathsf{F}$ is appended to $s$.
  \end{itemize}
\end{definition}

We can state that the number of recursive calls of $\textproc{Match}$ equals the number of $\mathsf{C}(q, i)$ that appeared in the corresponding run of $\textproc{Match}$; that is, the number of recursive calls of $\textproc{Match}$ is bounded by $O(|s|)$ with a run $s$ of $\textproc{Match}$.
Similarly, for a run $s$ of $\textproc{Memo}$, the number of recursive calls of $\textproc{Memo}$ equals the sum of the numbers of $\mathsf{C}(q, i)$ and $\mathsf{M}(q, i)$ that appeared in $s$ and is bounded by $O(|s|)$.

When a call $\mathsf{C}(q, i)$ appears in a run $s$ of $\textproc{Match}$, the return value ($\mathsf{F}$ or $(\mathsf{S}(i'))$) of that call is found later in $s$.
It is the same for a run of $\textproc{Memo}$.
Additionally, in a run of $\textproc{Memo}$, $\mathsf{M}(q, i)$ must be followed by $\mathsf{F}$ immediately.
We define $\phi_s$ as a function that finds the result index corresponding to $\mathsf{C}(q, i)$ in a run $s$.

\begin{definition}[a function that finds the result index of $\mathsf{C}(q, i)$]\label{def:find-result}
  For a run $s$ of $\textproc{Memo}$ and an index $\ell \in \{ 0, \dots, |s| - 1 \}$ such that $s[\ell] = \mathsf{C}(q, i)$, a function $\phi_s$ to find the index of the result value of the $s[\ell] = C(q, i)$ call is defined as follows
  \begin{align*}
    \phi_s(\ell) &= \phi_{s}(\ell, 0) \\
    \phi_s(\ell, n) &=
      \begin{cases}
        \ell & \text{if $n \le 1$ and $s[\ell] = \mathsf{F}$ or $\mathsf{S}(i')$} \\
        \phi_s(\ell + 1, n - 1) & \text{if $n > 1$ and $s[\ell] = \mathsf{F}$ or $\mathsf{S}(i')$} \\
        \phi_s(\ell + 1, n + 1) & \text{if $s[\ell] = \mathsf{C}(q, i)$ or $\mathsf{M}(q, i)$},
      \end{cases}
  \end{align*}
  We also define the same function $\phi_s(\ell)$ for a run $s$ of $\textproc{Match}$.
\end{definition}

We will show that for any $q \in Q$ and $i \in \{ 0, \dots, |w| \}$, no calls on the same $q$ and $i$ can occur during the execution of the matching for $q$ and $i$, i.e., the execution of \crefrange{line:memoization-trans-begin}{line:memoization-trans-end} of \cref{alg:memoization}.

\begin{lemma}\label{lem:no-eps-loop}
  For a run $s$ of $\textproc{Memo}$ and an index $\ell \in \{ 0, \dots, |s| - 1 \}$ such that $s[\ell] = \mathsf{C}(q, i)$, neither $\mathsf{C}(q, i)$ nor $\mathsf{M}(q, i)$ appear in the subsequence $s[\ell+1] \dots s[\phi_{s}(\ell)]$.
\end{lemma}

\begin{proof}
  Since an NFA has no $\varepsilon$-loop (\cref{assum:no-eps-loops}), the lemma holds.
  \qed
\end{proof}

We show that for some $q \in Q$ and $i \in \{ 0, \dots, |w| \}$, the execution of the transition for $q$ and $i$ is performed at most once.

\begin{lemma}\label{lem:exec-once}
  For a run $s$ of \textproc{Memo}, a state $q \in Q$, and a position $i \in \{ 0, \dots, |w| \}$, $\mathsf{C}(q, i)$ appears at most once in $s$.
\end{lemma}

\begin{proof}
  We fix $\ell \in \{ 0, \dots, |s| - 1 \}$ to be an index such that $s[\ell] = C(q, i)$.
  When $s[\phi_s(\ell)] = F$, the lemma holds in this case because the memoization table is updated after (\cref{line:memoization-record} of \cref{alg:memoization}).
  When $s[\phi_s(\ell)] = S$, for any $\ell < \ell'$, $s[\ell']$ must be $S$, so the lemma holds.
  \qed
\end{proof}

We define the \emph{in-degree} of the NFA transition function and use it to show that the number of the calls on the same $q$ and $i$ is at most the in-degree of $q$.

\begin{definition}[in-degree of NFA]\label{def:in-deg}
  For an NFA $\mathcal{A} = (Q, q_0, F, T)$ and a state $q \in Q$, the \emph{in-degree} of $q$, denoted by $\#_\mathsf{in}(q)$, is defined naturally by
  \begin{align*}
    \#_\mathsf{in}(q)
    =&\textstyle \left|\,\bigcup \{ q' \in Q\ |\ T(q') = \mathsf{Eps}(q) \}\,\right| + \left|\, \bigcup \{ q' \in Q\ |\ T(q') = \mathsf{Char}(\sigma, q) \} \,\right| \\
    +&\textstyle \left|\,\bigcup \{ q' \in Q\ |\ T(q') = \mathsf{Branch}(q, q'') \} \,\right| \\
    +&\textstyle \left|\, \bigcup \{ q' \in Q\ |\ T(q') = \mathsf{Branch}(q'', q) \} \,\right| \\
    +&\textstyle \textnormal{(\textbf{if} $q = q_0$ \textbf{then} 1 \textbf{else} 0)}
  \end{align*}
  We also define the \emph{maximum in-degree} of $\mathcal{A}$ as $\#_\mathsf{in}(\mathcal{A}) = \max_{q \in Q} \#_\mathsf{in}(q)$.
\end{definition}

\begin{lemma}\label{lem:memo-in}
  For a run $s$ of \textproc{Memo}, a state $q \in Q$, and a position $i \in \{ 0, \dots, |w| \}$, $\mathsf{M}(q, i)$ appears at most $\#_\mathsf{in}(q) - 1$ times in $s$.
\end{lemma}

\begin{proof}
  If $\mathsf{M}(q, i)$ appears $\#_\mathsf{in}(q)$ times in $s$, $\mathsf{C}(q', i')$ for some $q'$ and $i'$, where $q'$ is a predecessor of $q$, appears twice or more. Then, it violates \cref{lem:exec-once} and the lemma holds.
  \qed
\end{proof}

Now, we can show the line-time complexity of \cref{alg:memoization} (\cref{thm:algo-lin}).

\begin{proof}[for \cref{thm:algo-lin}]
  For an NFA $\mathcal{A} = (Q, q_0, F, T)$, an input string $w$, and a position $i \in \{ 0, \dots, |w| \}$, let $s$ be the run of $\Call{\MemoMzAw}{q_0, i}$, by \cref{lem:exec-once,lem:memo-in}, $|s|$ is bounded by $O(\#_\mathsf{in}(\mathcal{A}) \times |Q| \times |w|)$ and the theorem holds.
  \qed
\end{proof}

To show \cref{thm:algo-cor}, we define some functions that replace subsequences of a run of $\textproc{Memo}$ with subsequences of a run of $\textproc{Match}$.

\begin{definition}[$\mathsf{sub}$]\label{def:sub}
  For a run $s$ of $\textproc{Memo}$ and an index $\ell \in \{ 0, \dots, |s| -1 \}$ such that $s[\ell] = \mathsf{M}(q, i)$, we define $\mathsf{sub}(s, \ell)$ as follows
  \begin{equation*}
    \mathsf{sub}(s, \ell) = s[0] \dots s[\ell - 1]\,\mathsf{C}(q, i)\,s[\ell'+1] \dots s[\phi_s(\ell')] s[\ell+2] \dots s[|s| - 1]
  \end{equation*}
  where $\ell' \in \{ 0, \dots, |s| - 1 \}$ is the first index on $s$ such that $s[\ell'] = \mathsf{C}(q, i)$.
\end{definition}

\begin{definition}[$\mathsf{sub}^\ast$]\label{def:sub-star}
  For a run $s$ of $\textproc{Memo}$, we define $\mathsf{sub}^\ast(s)$ as follows
  \begin{enumerate}
    \item Find the index of the first $\mathsf{M}(q, i)$ for some $q$ and $i$.
    \item If such an index is not found, then it returns $s$.
    \item Let $\ell$ be the index to be found in step 1.
    \item $s \gets \mathsf{sub}(s, \ell)$.
    \item Back to step 1.
  \end{enumerate}
\end{definition}

$\mathsf{sub}$ and $\mathsf{sub}^\ast$ preserve the last result of the run.
Furthermore, we can show the following lemma about the value of $\mathsf{sub}^\ast(s)$.

\begin{lemma}\label{lem:sub-star}
  For an NFA $\mathcal{A} = (Q, q_0, F, T)$, an input string $w$, and a position $i \in \{ 0, \dots, |w| \}$, let $s$ be a run of $\Call{\MatchAw}{q_0, i}$ and $s'$ be a run of $\Call{\MemoMzAw}{q_0, i}$, $s = \mathsf{sub}^\ast(s')$.
\end{lemma}

\begin{proof}
  The subsequence of $s$ corresponding to the first call for some $q$ and $i$ is the same as that of $s'$ by \cref{alg:matching,alg:memoization}.
  Since \cref{alg:matching} is pure, the subsequence of $s$ corresponding to the second and subsequent calls is the same as the first one.
  $\textproc{sub}$ replaces $\mathsf{M}(q, i)$ with the subsequence of the first call, and $\mathsf{sub}^\ast$ repeats it until $\mathsf{M}(q,i)$ no longer exists, so the lemma holds.
  \qed
\end{proof}

Finally, we show the correctness of $\textproc{Memo}$ with respect to \cref{prob:regexParMatch} (\cref{thm:algo-cor}).

\begin{proof}[for \cref{thm:algo-cor}]
  For an NFA $\mathcal{A} = (Q, q_0, F, T)$, an input string $w$, and a position $i_0 \in \{ 0, \dots, |w| \}$, we fix $s$ be a run of $\Call{\MatchAw}{q_0, i_0}$ and $s'$ be a run of $\Call{\MemoMzAw}{q_0, i_0}$.
  By \cref{lem:sub-star}, the last element of $s$ and $s'$ must be the same (i.e., $s[|s|-1] = s'[|s'|-1]$).
  We can retrieve the result value of the calls from the last element in the following manner; if the last element is $\mathsf{F}$, the result value is $\mathsf{Failure}$; if the last element is $\mathsf{S}(i')$, the result value is $\mathsf{SuccessAt}(i')$.
  Therefore, the theorem holds.
\end{proof}

\subsection{Proofs for the Linear-time Complexity and Correctness of $\textproc{\MemoLa}$ from \cref{alg:la-memoization} (\cref{thm:la-lin} and \cref{thm:la-cor})}\label{appendix:memo-la-proof}

First, we define runs of both $\textproc{\MatchLaAt}$ and $\textproc{\MemoLa}$.

\begin{definition}[a run of $\textproc{\MatchLaAt}$]\label{def:match-la-at-run}
  For a (la,at)-NFA $\mathcal{A} = (P, Q, q_0, F, T)$, an input string $w$, we fix a set $C = \{ \mathsf{C}(q, i) \mid q \in P, i \in \{ 0, \dots, |w| \} \} \cup \{ \mathsf{S}(i') \mid i' \in \{ 0, \dots, |w| \} \} \cup \{ \mathsf{F} \}$.
  For a position $i_0 \in \{ 0, \dots, |w| \}$, we define a run $s$ of $\Call{\MatchLaAtAw}{q_0, i_0}$ as a sequence of $C$, which is crafted in the following manner:
  \begin{itemize}
    \item The initial value of $s$ is the empty sequence, and it grows on the execution of $\Call{\MatchLaAtAw}{q_0, i_0}$.
    \item When $\textproc{\MatchLaAt}$ is called with some $q \in P, i \in \{ 0, \dots, |w| \}$, $\mathsf{C}(q, i)$ is appended to $s$.
    \item When $\textproc{\MatchLaAt}$ returns $\mathsf{SuccessAt}(i')$ for some $i' \in \{ 0, \dots, |w| \}$, $\mathsf{S}(i')$ is appended to $s$.
    \item When $\textproc{\MatchLaAt}$ returns $\mathsf{Failure}$, $\mathsf{F}$ is appended to $s$.
  \end{itemize}
\end{definition}

\begin{definition}[a run of $\textproc{\MemoLa}$]\label{def:memo-la-run}
  For a la-NFA $\mathcal{A} = (P, Q, q_0, F, T)$, an input string $w$, we fix a set $C' = C \cup \{ \mathsf{M}(q, i) \mid q \in P, i \in \{ 0, \dots, |w| \} \} \cup \{ \mathsf{S}' \}$.
  For a position $i_0 \in \{ 0, \dots, |w| \}$, we define a run $s$ of $\Call{\MemoLaMzAw}{q_0, i_0}$ as a sequence of $C'$, which is crafted in the following manner:
  \begin{itemize}
    \item The initial value of $s$ is the empty sequence, and it grows on the execution of $\Call{\MemoLaMzAw}{q_0, i_0}$.
    \item When $\textproc{\MemoLa}$ is called with some $q \in P, i \in \{ 0, \dots, |w| \}$, but $\mathsf{M}(q, i)$ is not defined, $\mathsf{C}(q, i)$ is appended to $s$.
    \item When $\textproc{\MemoLa}$ is called with some $q \in P, i \in \{ 0, \dots, |w| \}$ and $\mathsf{M}(q, i)$ is defined, $\mathsf{M}(q, i)$ is appended to $s$.
    \item When $\textproc{\MemoLa}$ returns $\mathsf{SuccessAt}(i')$ for some $i' \in \{ 0, \dots, |w| \}$, $\mathsf{S}(i')$ is appended to $s$.
    \item When $\textproc{\MemoLa}$ returns $\mathsf{Success}$, $\mathsf{S}'$ is appended to $s$.
    \item When $\textproc{\MemoLa}$ returns $\mathsf{Failure}$, $\mathsf{F}$ is appended to $s$.
  \end{itemize}
\end{definition}

Similarly to \cref{def:match-run,def:memo-run}, the sum of the numbers of $\mathsf{C}(q, i)$ and $\mathsf{M}(q, i)$ in a run is the same as the number of the recursive calls.
Also, for a run $s$, the number of the recursive calls is bounded by $O(|s|)$.

We define the similar one to $\phi_s$ (\cref{def:find-result}) for runs of $\textproc{\MatchLaAt}$ and $\textproc{\MemoLa}$.

\begin{definition}[a function that finds the result index of $\mathsf{C}(q, i)$]\label{def:find-result-la}
  For a run $s$ of $\textproc{\MemoLa}$ and an index $\ell \in \{ 0, \dots, |s| - 1 \}$ such that $s[\ell] = \mathsf{C}(q, i)$, a function $\phi_s$ to find the index of the result value of the $s[\ell] = C(q, i)$ call is defined as follows
  \begin{align*}
    \phi_s(\ell) &= \phi_{s}(\ell, 0) \\
    \phi_s(\ell, n) &=
      \begin{cases}
        \ell & \text{if $n \le 1$ and $s[\ell] = \mathsf{F}$, $\mathsf{S}(i')$ or $\mathsf{S}'$} \\
        \phi_s(\ell + 1, n - 1) & \text{if $n > 1$ and $s[\ell] = \mathsf{F}$, $\mathsf{S}(i')$ or $\mathsf{S}'$} \\
        \phi_s(\ell + 1, n + 1) & \text{if $s[\ell] = \mathsf{C}(q, i)$ or $\mathsf{M}(q, i)$},
      \end{cases}
  \end{align*}
  We also define the same function $\phi_s(\ell)$ for a run $s$ of $\textproc{\MatchLaAt}$.
\end{definition}

Lemmas similar to \cref{lem:no-eps-loop,lem:exec-once} also hold for runs of $\textproc{\MemoLa}$ because of the same proofs.

\begin{lemma}\label{lem:no-eps-loop-la}
  For a run $s$ of $\textproc{\MemoLa}$ and an index $\ell \in \{ 0, \dots, |s| - 1 \}$ such that $s[\ell] = \mathsf{C}(q, i)$, neither $\mathsf{C}(q, i)$ nor $\mathsf{M}(q, i)$ appear in the subsequence $s[\ell+1] \dots s[\phi_{s}(\ell)]$.
\end{lemma}

\begin{lemma}\label{lem:exec-once-la}
  For a run $s$ of \textproc{\MemoLa}, a state $q \in P$, and a position $i \in \{ 0, \dots, |w| \}$, $\mathsf{C}(q, i)$ appears at most once in $s$.
\end{lemma}

We define the \emph{in-degree} of the (la,at)-NFA transition function and use it to show that the number of the calls on the same $q$ and $i$ is at most the in-degree of $q$.

\begin{definition}[in-degree of (la,at)NFA]\label{def:in-deg-la-at}
  For a (la,at)-NFA $\mathcal{A} = (P, Q, q_0, F, T)$ and a state $q \in P$, the \emph{in-degree} of $q$, denoted by $\#_\mathsf{in}(q)$, is defined naturally by
  \begin{align*}
    \#_\mathsf{in}(q)
    =&\textstyle \left|\,\bigcup \{ q' \in P\ |\ T(q') = \mathsf{Eps}(q) \}\,\right|
    +            \left|\,\bigcup \{ q' \in P\ |\ T(q') = \mathsf{Char}(\sigma, q) \} \,\right| \\
    +&\textstyle \left|\,\bigcup \{ q' \in P\ |\ T(q') = \mathsf{Branch}(q, q'') \} \,\right| \\
    +&\textstyle \left|\,\bigcup \{ q' \in P\ |\ T(q') = \mathsf{Branch}(q'', q) \} \,\right| \\
    +&\textstyle \left|\,\bigcup \{ q' \in P\ |\ T(q') = \mathsf{Sub}(k, \mathcal{A}', q) \} \,\right| \\
    +&\textstyle \textnormal{(\textbf{if} $q = q_0$ \textbf{then} 1 \textbf{else} 0) } \\
    +&\textstyle \textnormal{(\textbf{if} $q = {q_0}'$ \textbf{then} 1 \textbf{else} 0,}\\
     &\qquad \textnormal{for some sub-automaton $\mathcal{A}' = (P', Q', {q_0}', F', T')$)}
  \end{align*}
  We also define the \emph{maximum in-degree} of $\mathcal{A}$ as $\#_\mathsf{in}(\mathcal{A}) = \max_{q \in P} \#_\mathsf{in}(q)$.
\end{definition}

\begin{lemma}\label{lem:memo-in-la}
  For a run $s$ of \textproc{\MemoLa}, a state $q \in P$, and a position $i \in \{ 0, \dots, |w| \}$, $\mathsf{M}(q, i)$ appears at most $\#_\mathsf{in}(q) - 1$ times in $s$.
\end{lemma}

\begin{proof}
  This proof is the same as Lemma 5.
  Note that $\#_\mathsf{in}({q_0}')$ for the initial state ${q_0}'$ of each sub-automaton is greater than the other states by $1$.
  \qed
\end{proof}

Now, we can show the line-time complexity of \cref{alg:la-memoization} (\cref{thm:la-lin}).

\begin{proof}[for \cref{thm:la-lin}]
  For a la-NFA $\mathcal{A} = (P, Q, q_0, F, T)$, an input string $w$, and a position $i \in \{ 0, \dots, |w| \}$, let $s$ be the run of $\Call{\MemoLaMzAw}{q_0, i}$, by \cref{lem:exec-once-la,lem:memo-in-la}, $|s|$ is bounded by $O(\#_\mathsf{in}(\mathcal{A}) \times |Q| \times |w|)$ and the theorem holds.
  \qed
\end{proof}

To show \cref{thm:la-cor}, we define some functions that replace subsequences of a run of $\textproc{\MemoLa}$ with subsequences of a run of $\textproc{\MatchLaAt}$.

\begin{definition}[$\mathsf{sub\text{-}la}$]\label{def:sub-la}
  For a run $s$ of $\textproc{\MemoLa}$ and an index $\ell \in \{ 0, \dots, |s| -1 \}$ such that $s[\ell] = \mathsf{M}(q, i)$, we define $\mathsf{sub\text{-}la}(s, \ell)$ as follows
  \begin{equation*}
    \mathsf{sub\text{-}la}(s, \ell) = s[0] \dots s[\ell - 1]\,\mathsf{C}(q, i)\,s[\ell'+1] \dots s[\phi_s(\ell')] s[\ell+2] \dots s[|s| - 1]
  \end{equation*}
  where $\ell' \in \{ 0, \dots, |s| - 1 \}$ is the first index on $s$ such that $s[\ell'] = \mathsf{C}(q, i)$.
\end{definition}

\begin{definition}[$\mathsf{sub\text{-}la}^\ast$]\label{def:sub-la-star}
  For a run $s$ of $\textproc{\MemoLa}$, we define $\mathsf{sub\text{-}la}^\ast(s)$ as follows
  \begin{enumerate}
    \item Find the index of the first $\mathsf{M}(q, i)$ for some $q$ and $i$.
    \item If such an index is not found, then it returns $s$.
    \item Let $\ell$ be the index to be found in step 1.
    \item $s \gets \mathsf{sub\text{-}la}(s, \ell)$.
    \item Back to step 1.
  \end{enumerate}
\end{definition}

Similarly to \cref{def:sub,def:sub-star}, $\mathsf{sub\text{-}la}$ and $\mathsf{sub\text{-}la}^\ast$ preserve the last result of the run.
Furthermore, we can show the following lemma about the value of $\mathsf{sub\text{-}la}^\ast(s)$ for the same reason as \cref{lem:sub-star}

\begin{lemma}\label{lem:sub-la-star}
  For a la-NFA $\mathcal{A} = (P, Q, q_0, F, T)$, an input string $w$, and a position $i \in \{ 0, \dots, |w| \}$, let $s$ be a run of $\Call{\MatchLaAtAw}{q_0, i}$ and $s'$ be a run of $\Call{\MemoLaMzAw}{q_0, i}$, $s = \mathsf{sub\text{-}la}^\ast(s')$.
\end{lemma}

Finally, we show the correctness of $\textproc{\MemoLa}$ with respect to \cref{prob:laAtRegexMatch} (\cref{thm:la-cor}).

\begin{proof}[for \cref{thm:la-cor}]
  This proof is the same as for \cref{thm:algo-cor}.
\end{proof}

\subsection{Proofs for for the Linear-time Complexity and Correctness of $\textproc{\MemoAt}$ from \cref{alg:at-memoization} (\cref{thm:at-lin} and \cref{thm:at-cor})}\label{appendix:memo-at-proof}

First, we define a run for $\textproc{\MemoAt}$.
Note that a run of $\textproc{\MatchLaAt}$ is defined in \cref{appendix:memo-la-proof} (\cref{def:match-la-at-run}).

\begin{definition}[a run of $\textproc{\MemoAt}$]\label{def:memo-at-run}
  For an at-NFA $\mathcal{A} = (P, Q, q_0, F, T)$, an input string $w$, we fix a set $C' = \{ \mathsf{C}(q, i) \mid q \in P, i \in \{ 0, \dots, |w| \} \} \cup \{ \mathsf{S}(i') \mid i' \in \{ 0, \dots, |w| \} \} \cup \{ \mathsf{F}(j) \mid j \in \{ 0, \dots, \nu(\mathcal{A}) \} \} \cup \{ \mathsf{M}(q, i) \mid q \in P, i \in \{ 0, \dots, |w| \} \}$.
  For a position $i_0 \in \{ 0, \dots, |w| \}$, we define a run $s$ of $\Call{\MemoAtMzAAw}{q_0, i_0}$ as a sequence of $C'$, which is crafted in the following manner:
  \begin{itemize}
    \item The initial value of $s$ is the empty sequence, and it grows on the execution of $\Call{\MemoAtMzAAw}{q_0, i_0}$.
    \item When $\textproc{\MemoAt}$ is called with some $q \in P, i \in \{ 0, \dots, |w| \}$, but $\mathsf{M}(q, i)$ is not defined, $\mathsf{C}(q, i)$ is appended to $s$.
    \item When $\textproc{\MemoAt}$ is called with some $q \in P, i \in \{ 0, \dots, |w| \}$ and $\mathsf{M}(q, i)$ is defined, $\mathsf{M}(q, i)$ is appended to $s$.
    \item When $\textproc{\MemoAt}$ returns $\mathsf{SuccessAt}(i')$ for some $i' \in \{ 0, \dots, |w| \}$, $\mathsf{S}(i')$ is appended to $s$.
    \item When $\textproc{\MemoAt}$ returns $\mathsf{Failure}(j)$ for some $j \in \{ 0, \dots, \nu(\mathcal{A}) \}$, $\mathsf{F}(j)$ is appended to $s$.
  \end{itemize}
\end{definition}

Similarly to \cref{def:match-run,def:memo-run}, the sum of the numbers of $\mathsf{C}(q, i)$ and $\mathsf{M}(q, i)$ in a run is the same as the number of the recursive calls.
Also, for a run $s$ of $\textproc{\MemoAt}$, the number of the recursive calls is bounded by $O(|s|)$.

We define the similar one to $\phi_s$ (\cref{def:find-result}) for runs of $\textproc{\MemoAt}$.

\begin{definition}[a function that finds the result index of $\mathsf{C}(q, i)$]\label{def:find-result-at}
  For a run $s$ of $\textproc{\MemoAt}$ and an index $\ell \in \{ 0, \dots, |s| - 1 \}$ such that $s[\ell] = \mathsf{C}(q, i)$, a function $\phi_s$ to find the index of the result value of the $s[\ell] = C(q, i)$ call is defined as follows
  \begin{align*}
    \phi_s(\ell) &= \phi_{s}(\ell, 0) \\
    \phi_s(\ell, n) &=
      \begin{cases}
        \ell & \text{if $n \le 1$ and $s[\ell] = \mathsf{F}(j)$ or $\mathsf{S}(i')$} \\
        \phi_s(\ell + 1, n - 1) & \text{if $n > 1$ and $s[\ell] = \mathsf{F}(j)$ or $\mathsf{S}(i')$} \\
        \phi_s(\ell + 1, n + 1) & \text{if $s[\ell] = \mathsf{C}(q, i)$ or $\mathsf{M}(q, i)$},
      \end{cases}
  \end{align*}
\end{definition}


A lemma similar to \cref{lem:no-eps-loop} holds for runs of $\textproc{\MemoAt}$ because of the same proofs.

\begin{lemma}\label{lem:no-eps-loop-at}
  For a run $s$ of $\textproc{\MemoAt}$ and an index $\ell \in \{ 0, \dots, |s| - 1 \}$ such that $s[\ell] = \mathsf{C}(q, i)$, neither $\mathsf{C}(q, i)$ nor $\mathsf{M}(q, i)$ appear in the subsequence $s[\ell+1] \dots s[\phi_{s}(\ell)]$.
\end{lemma}

A lemma similar to \cref{lem:exec-once} also holds for a run of $\textproc{\MemoAt}$.

\begin{lemma}\label{lem:exec-once-at}
  For a run $s$ of \textproc{\MemoAt}, a state $q \in P$, and a position $i \in \{ 0, \dots, |w| \}$, $\mathsf{C}(q, i)$ appears at most once in $s$.
\end{lemma}

\begin{proof}
  We fix $\ell \in \{ 0, \dots, |s| - 1 \}$ to be an index such that $s[\ell] = C(q, i)$.
  When $s[\phi_s(\ell)] = F(j)$, the lemma holds in this case because the memoization table is updated after (Line \cref{line:at-memoization-update} of \cref{alg:at-memoization}).
  When $s[\phi_s(\ell)] = S$, we consider the case that $q$ is a state in an $\mathsf{at}$ sub-automaton because the other case is easy.
  In this case, $s[\phi_s(\ell)] = S$ means the call in \cref{line:at-memoization-call-1} returns $\mathsf{SuccessAt}(i', K')$, and $\Call{\MemoAtMAzAw}{q', i}$ is called in \cref{line:at-memoization-call-2}.
  Let $\ell'$ is an index such that $s[\ell'] = \mathsf{C}(q', i)$ or $\mathsf{M}(q', i)$ corresponding to this call in \cref{line:at-memoization-call-2}, by \cref{lem:no-eps-loop}, $C(q, i)$ cannot appear in $s[\ell'] \dots s[\phi(\ell')]$.
  When $s[\phi_s(\ell')] = \mathsf{F}$, \cref{line:at-memoization-batch} is performed and the lemma holds.
  When $s[\phi_s(\ell')] = \mathsf{S}$, because the same discussion for an ambient automaton is possible, the lemma also holds in this case.
  \qed
\end{proof}

A lemma similar to \cref{lem:memo-in-at} holds for runs of $\textproc{\MemoAt}$ because of the same proofs.

\begin{lemma}\label{lem:memo-in-at}
  For a run $s$ of \textproc{\MemoAt}, a state $q \in P$, and a position $i \in \{ 0, \dots, |w| \}$, $\mathsf{M}(q, i)$ appears at most $\#_\mathsf{in}(q) - 1$ times in $s$.
\end{lemma}

Now, we can show the line-time complexity of \cref{alg:at-memoization} (\cref{thm:at-lin}).

\begin{proof}[for \cref{thm:at-lin}]
  For an at-NFA $\mathcal{A} = (P, Q, q_0, F, T)$, an input string $w$, and a position $i \in \{ 0, \dots, |w| \}$, let $s$ be the run of $\Call{\MemoAtMzAAw}{q_0, i}$, by \cref{lem:exec-once-at,lem:memo-in-at}, $|s|$ is bounded by $O(\#_\mathsf{in}(\mathcal{A}) \times |Q| \times |w|)$ and the theorem holds.
  \qed
\end{proof}

To show \cref{thm:la-cor}, we first define a function to find an index corresponding to the result recorded in the memoization table.
That is, for an index $\ell$ such that $s[\ell] = \mathsf{C}(q, i)$, $s[\phi^\mathsf{C}_s(\ell)]$ must be $F(j)$ and $M(q, i) = \mathsf{Failure}(j)$ is recorded during the execution of $\textproc{\MemoAt}$.

\begin{definition}[a function that finds the memoized index of $\mathsf{C}(q, i)$]\label{def:find-memo-at}
  For an at-NFA $\mathcal{A} = (P, Q, q_0, F, T)$, and a run $s$ of $\textproc{\MemoAt}$ for $\mathcal{A}$ and an index $\ell \in \{ 0, \dots, |s| -1 \}$ such that $s[\ell] = \mathsf{C}(q, i)$, we define $\phi^\mathsf{C}_s(\ell)$ as follows
  \begin{equation*}
    \phi^\mathsf{C}_s(\ell) = \begin{cases*}
      \phi_s(\ell) & \text{if $s[\phi_s(\ell)] = \mathsf{F}(j)$} \\
      \phi^\mathsf{C}_s(\phi_s(\ell') + 1) & \text{if $s[\phi_s(\ell)] = \mathsf{S}$}
    \end{cases*}
  \end{equation*}
  where $\mathcal{A}' = (P', Q', {q_0}', F', T')$ such that $q \in Q'$, and $\ell' \in \{ 0, \dots, |s|-1 \}$ such that $s[\ell'] = \mathsf{C}({q_0}', i')$ and $\ell' < \ell < \phi(\ell')$.
\end{definition}

We define a function to find an index of the end of the range where a matching is skipped as a result of using the value recorded in the memoization table.
In the normal case, this is the next index, but if a failure with different nesting depths is recorded and used, sub-automata are skipped until the ambient automaton where the failure occurred.

\begin{definition}[a function that finds the skipped index of $\mathsf{M}(q, i)$]\label{def:find-skip-at}
  For an at-NFA $\mathcal{A} = (P, Q, q_0, F, T)$, and a run $s$ of $\textproc{\MemoAt}$ for $\mathcal{A}$ and an index $\ell \in \{ 0, \dots, |s| -1 \}$ such that $s[\ell] = \mathsf{M}(q, i)$ and $s[\ell + 1] = F(j)$, we define $\phi^\mathsf{M}_s(\ell)$ as follows
  \begin{align*}
    \phi^\mathsf{M}_s(\ell) &= \begin{cases*}
      \ell + 1 & \text{$j = \nu_\mathcal{A}(q)$} \\
      \downarrow^\mathsf{M}_s(\ell) & \text{$j < \nu_\mathcal{A}(q)$}
    \end{cases*} \\
    \downarrow^\mathsf{M}_s(\ell') &= \begin{cases*}
      \phi_s(\ell'') & \text{$j + 1 = \nu_\mathcal{A}({q_0}')$} \\
      \downarrow^\mathsf{M}_s(\ell'' - 1) & \text{$j + 1 < \nu_\mathcal{A}({q_0}')$}
    \end{cases*}
  \end{align*}
  where $\mathcal{A}' = (P', Q', {q_0}', F', T')$ such that $q' \in Q'$ for $s[\ell'] = \mathsf{C}(q', i')$ or $\mathsf{M}(q', i')$, and $\ell'' \in \{ 0, \dots, |s|-1 \}$ such that $s[\ell''] = \mathsf{C}({q_0}', i'')$ and $\ell'' < \ell' < \phi_s(\ell'')$.
\end{definition}

We define some functions that replace subsequences of a run of $\textproc{\MemoAt}$ with subsequences of a run of $\textproc{\MatchLaAt}$.

\begin{definition}[$\mathsf{sub\text{-}at}$]\label{def:sub-at}
  For a run $s$ of $\textproc{\MemoAt}$ and an index $\ell \in \{ 0, \dots, |s| -1 \}$ such that $s[\ell] = \mathsf{M}(q, i)$, we define $\mathsf{sub\text{-}at}(s, \ell)$ as follows
  \begin{equation*}
    \mathsf{sub\text{-}la}(s, \ell) = s[0] \dots s[\ell - 1]\,\mathsf{C}(q, i)\,s[\ell'+1] \dots s[\phi^\mathsf{C}_s(\ell')] s[\phi^\mathsf{M}_s(\ell)+1] \dots s[|s| - 1]
  \end{equation*}
  where $\ell' \in \{ 0, \dots, |s| - 1 \}$ is the first index on $s$ such that $s[\ell'] = \mathsf{C}(q, i)$.
\end{definition}

\begin{definition}[$\mathsf{sub\text{-}at}^\ast$]\label{def:sub-at-star}
  For a run $s$ of $\textproc{\MemoAt}$, we define $\mathsf{sub\text{-}at}^\ast(s)$ as follows
  \begin{enumerate}
    \item Find the index of the first $\mathsf{M}(q, i)$ for some $q$ and $i$.
    \item If such an index is not found, then it returns $s$.
    \item Let $\ell$ be the index to be found in step 1.
    \item $s \gets \mathsf{sub\text{-}at}(s, \ell)$.
    \item Back to step 1.
  \end{enumerate}
\end{definition}

Similarly to \cref{def:sub,def:sub-star}, $\mathsf{sub\text{-}at}$ and $\mathsf{sub\text{-}at}^\ast$ preserve the last result of the run.
Furthermore, we can show the following lemma about the value of $\mathsf{sub\text{-}at}^\ast(s)$ for the same reason as \cref{lem:sub-star}.
Here $g$ is a function convert on types from runs of $\textproc{Memo}$ to runs of $\textproc{Match}$, which is defined as $g(\varepsilon) = \varepsilon,\,g(F(j)\,s_1 \dots s_n) = F\,g(s_1 \dots s_n)$ and $g(s_0\,s_1 \dots s_n) = s_0\,g(s_1 \dots s_n)$.

\begin{lemma}\label{lem:sub-la-star}
  For a la-NFA $\mathcal{A} = (P, Q, q_0, F, T)$, an input string $w$, and a position $i \in \{ 0, \dots, |w| \}$, let $s$ be a run of $\Call{\MatchLaAtAw}{q_0, i}$ and $s'$ be a run of $\Call{\MemoAtMzAAw}{q_0, i}$, $s = g(\mathsf{sub\text{-}at}^\ast(s'))$.
\end{lemma}

Finally, we show the correctness of $\textproc{\MemoAt}$ with respect to \cref{prob:laAtRegexMatch} (\cref{thm:at-cor}).

\begin{proof}[for \cref{thm:at-cor}]
  This proof is the same as for \cref{thm:algo-cor}.
\end{proof}

\section{Memoization for Regexes with Look-around and Atomic Grouping}\label{appendix:la-at-memoization}

\cref{alg:la-at-memoization} is our matching algorithm for (la, at)-NFAs; it is a combination of \cref{alg:la-memoization} and \cref{alg:at-memoization}.

\cref{alg:la-at-memoization} exhibits the desired properties, namely correctness (with respect to \cref{prob:laAtRegexMatch}) and linear-time complexity.

\begin{theorem}[linear-time complexity of \cref{alg:la-at-memoization}]\label{thm:la-at-lin}
  For an NFA $\mathcal{A} = (Q, q_0, F, T)$, an input string $w$, and an position $i \in \{ 0, \dots, |w| \}$, $\Call{\MemoLaAtMzAAw}{q_0, i}$ terminates with $O(|w|)$ recursive calls.
\end{theorem}

\begin{theorem}[correctness of \cref{alg:la-at-memoization}]\label{thm:la-at-cor}
  For an NFA $\mathcal{A} = (Q, q_0, F, T)$, an input string $w$, and an position $i \in \{ 0, \dots, |w| \}$, $\Call{\MatchLaAtAw}{q_0, i} = f(\Call{\MemoLaAtMzAAw}{q_0, i})$.
\end{theorem}

We can also show \cref{thm:la-at-lin,thm:la-at-cor} by using the strategy combined with \cref{thm:la-lin,thm:at-lin} (for linear-time complexity, \cref{thm:la-at-lin}) or \cref{thm:la-cor,thm:at-cor} (for correctness, \cref{thm:la-at-cor}).
It should be noted that the nesting depth of atomic grouping is reset on an $\mathsf{at}$ transition (Line 17).

\begin{algorithm}[tbp]
  \caption{our matching algorithm with memoization for (la,at)-NFAs}\label{alg:la-at-memoization}
  \begin{algorithmic}[1]
    \Function{\MemoLaAtMAzAw}{$q, i$}
      \Statex\hspace*{\algorithmicindent}
      \begin{tabular}{rl} 
       \textbf{Parameters}: & a (la,at)-NFA $\mathcal{A}_0$, a sub-automaton $\mathcal{A}$ of $\mathcal{A}_0$ (it can be $\mathcal{A}_0$ itself), \\
                            &an input string $w$, and a memoization table $M\colon P \times \mathbb{N} \rightharpoonup$ \\
                            & \quad\qquad $\{ \mathsf{Failure}(j) \mid j \in \{ 0, \dots, \nu(\mathcal{A}_0) \} \} \cup \{ \mathsf{Success} \}$ \\
            \textbf{Input}: & a current state $q$, and a current position $i$ \\
           \textbf{Output}: & returns $\mathsf{SuccessAt}(i', K)$ if the matching succeeds, \\
                            & returns $\mathsf{Success}$ if the matching succeeds with memoization, or \\
                            & returns $\mathsf{Failure}(j)$ if the matching fails
      \end{tabular}
      \State$(P, Q, q_0, F, T) = \mathcal{A}$
      \If{$M(q, i) \ne \bot$}
        \Return$M(q, i)$
      \EndIf%
      \State$\mathsf{result} \gets \bot$
      \If{$q \in F$}
        $\mathsf{result} \gets \mathsf{Success}(i, \emptyset)$
      \ElsIf{$T(q) = \mathsf{Eps}(q')$}
        $\mathsf{result} \gets \Call{\MemoLaAtMAzAw}{q', i}$
      \ElsIf{$T(q) = \mathsf{Branch}(q', q'')$}
        \State$\mathsf{result} \gets \Call{\MemoLaAtMAzAw}{q', i}$
        \If{$\mathsf{result} = \mathsf{Failure}(j)$ \textbf{and} $j = \nu_{\mathcal{A}_0}(q)$}
          \State$\mathsf{result} \gets \Call{\MemoLaAtMAzAw}{q'', i}$
          \If{$\mathsf{result} = \mathsf{Failure}(j')$}
            $\mathsf{result} \gets \mathsf{Failure}(\min(j, j'))$
          \EndIf%
        \EndIf%
      \ElsIf{$T(q) = \mathsf{Char}(\sigma, q')$}
        \If{$i < |w|$ \textbf{and} $w[i] = \sigma$}
          $\mathsf{result} \gets \Call{\MemoLaAtMAzAw}{q', i + 1}$
        \Else\ 
          $\mathsf{result} \gets \mathsf{Failure}(\nu_{\mathcal{A}_0}(q))$
        \EndIf%
      \ElsIf{$T(q) = \mathsf{Sub}(\mathsf{pla}, \mathcal{A}', q')$}
        \State$(P', Q', {q_0}', F', T') = \mathcal{A}'$
        \State$\mathsf{result} \gets \Call{\MemoLaAtMApApw}{q_0', i}$
        \If{$\mathsf{result} = \mathsf{SuccessAt}(i', K')$}
          \For{$k \in K$'}
            $M(k) \gets \mathsf{Success}$
          \EndFor%
          \State$\mathsf{result} \gets \Call{\MemoLaAtMAzAw}{q', i}$
        \ElsIf{$\mathsf{result} = \mathsf{Success}$}
          $\mathsf{result} \gets \Call{\MemoLaAtMAzAw}{q', i}$
        \ElsIf{$\mathsf{result} = \mathsf{Failure}(j)$}
          $\mathsf{result} \gets \mathsf{Failure}(\nu_{\mathcal{A}_0}(q))$
        \EndIf%
      \ElsIf{$T(q) = \mathsf{Sub}(\mathsf{at}, \mathcal{A}', q')$}
        \State$(P', Q', {q_0}', F', T') = \mathcal{A}'$
        \State$\mathsf{result} \gets \Call{\MemoLaAtMAzApw}{q_0', i}$
        \If{$\mathsf{result} = \mathsf{SuccessAt}(i', K')$}
          \State$\mathsf{result} \gets \Call{\MemoLaAtMAzAw}{q', i'}$
          \If{$\mathsf{result} = \mathsf{SuccessAt}(i'', K'')$}
            $\mathsf{result} \gets \mathsf{SuccessAt}(i'', K' \cup K'')$
          \ElsIf{$\mathsf{result} = \mathsf{Success}$}
            \For{$k \in K'$}
              $M(k) \gets \mathsf{Success}$
            \EndFor%
          \ElsIf{$\mathsf{result} = \mathsf{Failure}(j)$}
            \For{$k \in K'$}
              $M(k) \gets \mathsf{Failure}(j)$
            \EndFor%
          \EndIf%
        \ElsIf{$\mathsf{result} = \mathsf{Failure}(j)$ \textbf{and} $j > \nu_{\mathcal{A}_0}(q)$}
          $\mathsf{result} \gets \mathsf{Failure}(\nu_{\mathcal{A}_0}(q))$
        \EndIf%
      \EndIf%
      \If{$\mathsf{result} = \mathsf{SuccessAt}(i', K)$}
        $\mathsf{result} \gets \mathsf{SuccessAt}(i', K \cup \{ (q, i) \})$
      \ElsIf{$\mathsf{result} = \mathsf{Success}$}\ 
        $M(q, i) \gets \mathsf{Success}$
      \ElsIf{$\mathsf{result} = \mathsf{Failure}(j)$}\ 
        $M(q, i) \gets \mathsf{Failure}(j)$
      \EndIf%
      \State\Return$\mathsf{result}$
    \EndFunction%
  \end{algorithmic}
\end{algorithm}

\end{ArxivVersion}
\end{document}